\documentclass[11pt]{article}

\usepackage{amssymb,amsmath,amsfonts,amsthm,graphicx}

\def\be{\begin{equation}}
\def\ee{\end{equation}}
\def\bea{\begin{eqnarray}}
\def\eea{\end{eqnarray}}

\newcommand{\sect}[1]{\setcounter{equation}{0}\section{#1}}
\newcommand{\subsect}[1]{\subsection{#1}}

\theoremstyle{plain}
\newtheorem{theorem}{Theorem}

\newtheorem{lemma}[theorem]{Lemma}
\newtheorem{proposition}[theorem]{Proposition}

\theoremstyle{definition}
 \newtheorem{definition}[theorem]{Definition}
\newtheorem{note}[theorem]{Note}
\newtheorem{example}[theorem]{Example}

\def\dd{{\rm d}}
 \def\ota{\iota}
 \def\iii{{\cal I}}

 \def\be{\begin{equation}}
\def\ee{\end{equation}}
\def\bea{\begin{eqnarray}}
\def\eea{\end{eqnarray}}

\numberwithin{theorem}{section}

\begin{document}

\thispagestyle{empty}

\medskip
\medskip

 \vskip2cm

 \begin{center}

\noindent {\Large{\bf {Lie--Hamilton systems on the plane:}}}\\
\vskip0.4cm
\noindent {\Large{\bf Applications and superposition rules}}

\medskip
\medskip
\medskip

{\sc Alfonso Blasco$^{1}$, Francisco J Herranz$^{1}$, Javier de Lucas$^{2}$ and Cristina Sard\'on$^{3}$}

 \end{center}

\medskip

{\small
\noindent
$^1$ Department of Physics, University of Burgos,   09001, Burgos, Spain

\noindent
$^2$ Department of Mathematical Methods in Physics, University of Warsaw, Pasteura 5, 02-093, Warszawa, Poland

\noindent
$^3$ Department of Fundamental Physics, University of Salamanca, Plaza de la Merced s/n, 37008, Salamanca, Spain
}
\medskip

\noindent
{\small 
 E-mail: {\tt   ablasco@ubu.es, fjherranz@ubu.es, javier.de.lucas@fuw.edu.pl,\\ cristinasardon@usal.es}}

  \medskip
\bigskip
\bigskip

\begin{abstract}
\noindent A Lie--Hamilton system is a nonautonomous system of first-order ordinary differential equations describing the integral curves of a $t$-dependent vector field taking values in a finite-dimensional real Lie algebra of Hamiltonian vector fields with respect to a Poisson structure. We provide new algebraic/geometric techniques to easily determine the properties of such Lie algebras on the plane, e.g., their associated Poisson bivectors. We study new and known Lie--Hamilton systems on $\mathbb{R}^2$ with physical, biological and mathematical applications. New results cover Cayley--Klein Riccati equations, the here defined planar diffusion Riccati systems, complex Bernoulli differential equations and projective Schr\"odinger equations. Constants of motion  for   planar Lie--Hamilton systems are explicitly obtained which, in turn, allow us to derive   superposition rules through a coalgebra approach.

 \end{abstract}

\bigskip\bigskip\bigskip\bigskip

\noindent
MSC:  {34A26 (primary) 70G45, 70H99 (secondary)}

\medskip

\noindent{KEYWORDS}:  complex Bernoulli differential equation, Casimir element, constant of motion, Lie system,  Poisson coalgebra,  superposition rule, symplectic geometry, Vessiot--Guldberg Lie algebra
 \newpage


\sect{Introduction}

A {\em  Lie system} is a nonautonomous system of first-order ordinary differential equations whose general solution can be written as a function, a {\em superposition rule},
of a generic family of particular solutions and some constants related to initial conditions~\cite{LS,PW,CGM00,Dissertations}. Superposition rules significantly simplify the determination of general solutions for Lie systems as they reduce its derivation to obtaining several particular solutions.

Examples of Lie systems are linear systems of differential equations, Riccati equations and matrix Riccati equations~\cite{Dissertations,ORW05,Abou}. 
Most nonautonomous systems of first-order ordinary differential equations are not Lie systems~\cite{Dissertations,In72,CGL10}. Notwithstanding,
Lie systems play a very significant r\^ole due to their relevant applications and geometric properties, e.g.,~matrix Riccati equations are ubiquitous in control theory and superposition rules can be understood as a particular type of flat connection~\cite{Dissertations,Abou,CGM07}.

Choose global coordinates $\{x,y\}$ on $\mathbb{R}^2$. To study Lie systems on $\mathbb{R}^2$, we relate every nonautonomous system
\begin{equation}\label{system}
 \frac{{\rm d} x}{{\rm d} t  }=f(t,x,y), \qquad \frac{{\rm d} y}{{\rm d} t }=g(t,x,y),
\end{equation}
where $f,g:\mathbb{R}^3\rightarrow \mathbb{R}$ are arbitrary functions, to the $t$-dependent vector field 
\begin{equation}\label{Vect}
X:\mathbb{R}\times\mathbb{R}^2\ni (t,x,y)\mapsto f(t,x,y)\frac{\partial}{\partial x}+g(t,x,y)\frac{\partial}{\partial y}
\in {\rm T}\mathbb{R}^2
\end{equation}
 and vice versa. This permits us to use $X$ to refer to both (\ref{system}) and (\ref{Vect}). 
The {\em Lie--Scheffers Theorem}~\cite{LS,CGM00} states that $X$ is a Lie system if and only if 
$$
X_t(x,y):= X(t,x,y)=\sum_{i=1}^l b_i(t)X_i(x,y)
$$
  for some $t$-dependent functions $b_1(t),\ldots,b_l(t)$ and vector fields $X_1,\ldots,X_l$ on $\mathbb{R}^2$ spanning an $l$-dimensional real Lie algebra $V$ of vector fields: a {\em Vessiot--Guldberg Lie algebra} of $X$. If $V$ is isomorphic to a Lie algebra $\mathfrak{g}$, we say that $X$ is a {\it $\mathfrak{g}$-Lie system}.

As an example of Lie systems,  let us consider for the first time the family of nonautonomous complex Bernoulli differential equations \cite{Ma11} of the form
\begin{equation}\label{complexBernoulli}
 \frac{\dd z}{\dd t}=a_1(t)z+a_2(t)z^n,\qquad n\notin\{0,1\},
\end{equation}
where $z\in \mathbb{C}$ and $a_1(t),a_2(t)$ are arbitrary complex valued $t$-dependent functions. 
 If $a_1(t)$ and $a_2(t)$ are polynomial functions of ${\rm e}^{{\rm i}t}$ and ${\rm e}^{-{\rm i}t}$, then (\ref{complexBernoulli}) becomes a particular case of
the differential equations $ \frac{{\rm d}z}{{\rm d}t}=P(z,{\rm e}^{{\rm i}t},{\rm e}^{-{\rm i}t})$, where $z\in\mathbb{C}$ and $P$ is any polynomial function of their arguments. The number of periodic particular solutions of these latter equations has been studied in \cite{Zo00}.
Writing  $z=r{\rm e}^{{\rm i}\theta}$, $a_1(t)=a_1^R(t)+{\rm i}a_1^I(t)$ and $a_2(t)=a_2^R(t)+{\rm i}a^I_2(t)$ for real $t$-dependent functions $a^R_1(t),a_1^I(t),a_2^R(t),a_2^I(t)$, system (\ref{complexBernoulli}) becomes
\begin{equation}\label{PlanBer}
\begin{aligned}
\frac{{\rm d}r}{{\rm d}t}&=a^R_1(t)\,r+a_2^R(t)\,r^n\cos[\theta(n-1)]-a_2^I(t)\,r^ n\sin [\theta(n-1)],\\
\frac{{\rm d}\theta}{{\rm d}t}&=a^I_1(t)+a_2^R(t)\,r^{n-1}\sin[\theta(n-1)]+a_2^I(t)\,r^{n-1}\cos[\theta(n-1)],
\end{aligned}
\end{equation}
which is related to $X=a^R_1(t)X_0+a^I_1(t)X_1+a^R_2(t)X_2+a^I_2(t)X_3$, where
\begin{equation}\label{VectComBer}
\begin{gathered}
X_0=r\frac{\partial}{\partial r},\qquad X_1=\frac{\partial}{\partial \theta},\qquad X_2=r^n\cos[\theta(n-1)]\frac{\partial}{\partial r}+r^{n-1}\sin[\theta(n-1)]\frac{\partial}{\partial \theta},\\ X_3=-r^n\sin[\theta(n-1)]\frac{\partial}{\partial r}+r^{n-1}\cos[\theta(n-1)]\frac{\partial}{\partial \theta}
\end{gathered}
\end{equation}
span a four-dimensional real Lie algebra, $V^{\rm CB}$, with commutation relations
\begin{equation}\label{RelBer}
\begin{aligned}
\left[X_0,X_1\right]&=0,\qquad &[X_0,X_2]&=(n-1)X_2,\qquad &[X_0,X_3]&=(n-1)X_3,\\
[X_1,X_2]&=(n-1)X_3,\qquad &[X_1,X_3]&=-(n-1)X_2,\qquad &[X_2,X_3]&=0.
\end{aligned}
\end{equation}
So, $X$ takes values in the finite-dimensional Lie algebra $V^{\rm CB}$  and becomes a Lie system.

The Lie--Scheffers Theorem implies that classifying Lie systems on a fixed manifold amounts to determining all non-diffeomorphic finite-dimensional
real Lie algebras of vector fields on it \cite{GKP92,BBHLS}.
In the XIXth century, Lie accomplished the local classification of such Lie algebras on the plane~\cite{Li93,HA75}.
Gonz\'alez-L\'opez, Kamran and Olver reviewed Lie's classification using modern geometric techniques~\cite{GKP92}.
Their (GKO) classification  divides finite-dimensional real Lie algebras of planar vector fields  into $28$ non-diffeomorphic classes, which can be regarded as a local classification of Lie systems on $\mathbb{R}^2$ \cite{GKP92,BBHLS}. For instance, we see from (\ref{RelBer}) that system (\ref{PlanBer}) 
 admits a Vessiot--Guldberg Lie algebra $V^{\rm CB}\simeq \mathbb{R}^2\ltimes \mathbb{R}^2\simeq \langle X_0,X_1\rangle \ltimes \langle X_2, X_3\rangle $.
 According to the GKO classification, all Lie algebras of this type are locally diffeomorphic to the Lie algebra P$_4$  (cf.~\cite{GKP92,BBHLS}). That is why we say that (\ref{PlanBer}) is {\it a Lie system of class {\rm P}$_4$}. 

In this work we focus on studying planar Lie systems possessing a Vessiot--Guldberg Lie algebra of Hamiltonian vector fields with respect to a Poisson structure. It was shown in \cite[lemma 4.1]{BBHLS} that this essentially amounts to having a Lie algebra of Hamiltonian vector fields relative to a symplectic structure, a herafter called {\it compatible symplectic structure}. Among the $28$ classes of the GKO classification,  only  $12$ consist of Lie algebras of Hamiltonian vector fields. 
Table~\ref{table1} (see section 2) summarizes the classification of  finite-dimensional Lie algebras of Hamiltonian vector fields on $\mathbb{R}^2$ given in \cite{BBHLS}. 

There is not a Lie algebra   isomorphic to $V^{\rm CB}$  in table~\ref{table1}. Consequently, $V^{\rm CB}$ cannot be a Lie algebra of Hamiltonian vector fields relative to any Poisson structure. This illustrates that not every Lie system admits a Vessiot--Guldberg Lie algebra of Hamiltonian vector fields with respect to a Poisson structure \cite{IV,CMM97}. When a Lie system does, we call it a  {\em Lie--Hamilton (LH) system} \cite{CLS12Ham}. The interest of LH systems relies on their appearance in physics, mathematics and biology (cf.~\cite{BBHLS,CLS12Ham,BCHLS13Ham}). Additionally, their associated Poisson structures allow us to study and to derive their constants of motion, superposition rules and  Lie symmetries \cite{CLS12Ham,BCHLS13Ham}. 

As an example, consider the complex Bernoulli differential equations (\ref{PlanBer}) with $a_1^R(t)=0$ and the Poisson bivector 
\begin{equation}\label{BP}
\Lambda=X_2\wedge X_3=r^{2n-1}\frac{\partial}{\partial r}\wedge \frac{\partial}{\partial \theta}
\end{equation}
turning the elements of $V=\langle X_1,X_2,X_3\rangle$ into Hamiltonian vector fields. Indeed, the Hamiltonian functions for $X_1$, $X_2$, $X_3$  read
\be
h_1=\frac1{(2n-2)r^{2n-2}},\qquad h_2=\frac{\sin [\theta(n-1)]}{r^{n-1}(n-1)},\qquad 	h_3=\frac{\cos [\theta(n-1)]}{r^{n-1}(n-1)},
\label{zxa}
\ee
correspondingly. These functions along with $h_0=1$ fulfill
\be
\{h_1,h_2\}_{\Lambda}=-(n-1) h_3 ,\quad\  \{h_1,h_3\}_{\Lambda}=(n-1) h_2 ,\quad\  \{h_2,h_3\}_{\Lambda}= h_0,\quad\  \{h_0,\,\cdot\,\}_{\Lambda}= 0 .
\label{za}
\ee
Hence, system (\ref{PlanBer}) with $a_1^R(t)=0$ is a LH system as it is related to a $t$-dependent vector field taking values in a Vessiot--Guldberg Lie algebra $V$ of Hamiltonian vector fields relative to $\Lambda$. Since $V\simeq \mathbb{R}\ltimes \mathbb{R}^2\simeq \langle X_1\rangle\ltimes \langle X_2,X_3\rangle$, $X_1\wedge X_2\neq 0$ and ${\rm ad}_{X_1}: X_i\in \langle X_2,X_3\rangle\mapsto [X_1,X_i]\in  \langle X_2,X_3\rangle$ is not diagonalizable over $\mathbb{R}$, we see in view of table~\ref{table1} that the Lie algebra $V$ belongs to class P$_1$ and $V\simeq{\mathfrak{iso}}(2)$.  Meanwhile, the LH algebra spanned by $h_1,h_2,h_3,h_0$ is     isomorphic to the centrally extended  Euclidean algebra $\overline{\mathfrak{iso}}(2)$ (see also \cite{BBHLS} for further details).

On the other hand, we recall that 
superposition rules for LH systems can be obtained in an algebraic way by  applying a Poisson coalgebra approach  \cite{BCHLS13Ham}.
In contrast, other methods to derive superposition rules require to integrate a Vessiot--Guldberg Lie algebra, e.g.,~the group theoretical method \cite{PW}, or to solve a family of PDEs \cite{CGM07}. Winternitz and coworkers have also derived superposition rules for Lie systems in particular forms \cite{PW,ORW05}. The application of this latter result for general Lie systems requires to map them into the canonical form for which the superposition rule was obtained. The coalgebra procedure   makes these transformations unnecessary in many cases.

The structure of the paper is as follows. In section 2 we summarize the local classification of Vessiot--Guldberg Lie algebras as of Hamiltonian vector fields on the plane performed in~\cite{BBHLS}, where  the corresponding symplectic structures  were derived by solving a system of PDEs. As a first new achievement,  we show in section 3 that such symplectic structures can be determined through algebraic and geometric methods. Although our techniques are heavily based upon the Lie algebra structure of Vessiot--Guldberg Lie algebras, they also depend on their geometric properties as Lie algebras of vector fields.
In particular, a new method to construct symplectic structures for LH systems on the plane related to non-simple Vessiot--Guldberg Lie algebras is described.    

Next we remark that a Vessiot--Guldberg Lie algebra on the plane isomorphic to $\mathfrak{sl}(2)$ can be diffeomorphic to either P$_2$, I$_3$, I$_4$ or I$_5$, being I$_3$ the only class which does not consist of Hamiltonian vector fields for any Poisson bivector on $\mathbb{R}^2$ (cf. \cite{GKP92}).  The problem of determining the existence of diffeomorphisms among $\mathfrak{sl}(2)$-Lie systems on the plane is addressed 
 in section 4. As a second theoretical new result, we provide  a new easily verifiable algebraic-geometric criterium allowing one to determine the specific class of a Vessiot--Guldberg Lie algebra on $\mathbb{R}^2$ isomorphic to $\mathfrak{sl}(2)$ without finding a change of variables mapping it to a specific class as in~\cite{BBHLS}. Our new method is based on calculating an easily derivable geometric invariant:  a hereafter called {\it Casimir tensor field}.

 To illustrate the statements of section 4,  we retrieve some of the findings in \cite{BBHLS} and classify other $\mathfrak{sl}(2)$-LH systems on $\mathbb{R}^2$. More specifically, we show  in section 5 for the first time that Cayley--Klein Riccati equations~\cite{Liesym} comprise the three classes of   $\mathfrak{sl}(2)$-LH systems and we study in section 6 their relationships via diffeomorphisms to coupled Riccati equations~\cite{Mariton}, Milne--Pinney equations~\cite{LA08}, second-order Kummer--Schwarz equations~\cite{CGL11} and planar diffusion Riccati systems~\cite{SSVG11, SSVG14}.

In section 7,  we demonstrate that systems described by $t$-dependent quadratic Hamiltonians, e.g. $t$-dependent  damped harmonic oscillators or particles under certain electric fields~\cite{CR03,UY02}, and the second-order Riccati equations in Hamiltonian form~\cite{CLS12, CRS05} can be analyzed through LH systems of class P$_5$ and they become equivalent up to a local diffeomorphism. Such systems admit a $t$-dependent Hamiltonian function taking values in a Lie algebra of Hamiltonian functions isomorphic to the referred to as {\it two-photon algebra}~\cite{BBF09}.

Certain projective Schr\"odinger equations are studied in section 8 where we prove that   they belong to the class 
P$_3$, that is, they are  ${\mathfrak{so}}(3)$-LH systems. This reveals the interest of LH systems in geometric quantum mechanics.

Additionally, we analyze $\mathfrak{h}_2$-LH systems  in section~9, where $\mathfrak{h}_2$ stands for the two-dimensional Heisenberg Lie algebra~\cite{BBHLS} which arises within the class  I$_{14A}$ for ${r=1}$. As new results we prove that complex Bernoulli differential equations (\ref{complexBernoulli}) with $t$-dependent real coefficients belong to this class  and, consequently,  we   establish their  equivalence with  generalized Buchdahl equations~\cite{Bu64, CSL05, CN10}, appearing in Relativity,  and $t$-dependent Lotka--Volterra systems~\cite{Tr96, JHL05}, occurring in biology.

  In section 10 we obtain several superposition rules for LH systems by following the Poisson coalgebra approach  \cite{BCHLS13Ham}. With this aim, we, firstly, obtain $t$-independent constants of motion for   LH-systems through Casimir invariants (cf.~table~\ref{table2}). And, secondly, we   use   them in order to construct   superposition rules  for  LH systems of classes P$_1$, P$_5$ and I$^{r=1}_{14A}$ that were no considered in~\cite{BCHLS13Ham}.

Finally, a summary of all the specific LH systems considered throughout the paper (cf.~table~\ref{table3}) as well as some  open problems are addressed in section 11.


\sect{Local classification of LH systems on the
plane}

In general, we hereafter assume all structures to be smooth and globally defined. We also consider points where each Poisson bivector $\Lambda$ has locally constant rank. This simplifies the presentation and highlights our main results.  

A {\em generic point} of a Lie algebra $V$ of vector fields is a point around which the vector fields of $V$ span a regular distribution. We write ${\rm dom}\,V$ for the set of generic points of $V$. Every Lie algebra of planar vector fields is locally diffeomorphic around generic points to one of the 28 classes of vector fields of the GKO classification, which covers two subclasses called {\em primitive} (8 cases P$_x$)  and {\em imprimitive}  (20 cases  I$_x$) ones~\cite{GKP92}. 
 
 To determine which of the 28 classes can be considered as Vessiot--Guldberg Lie algebras of Hamiltonian vector fields,  a symplectic form
$$
\omega=f(x,y){\rm d}x\wedge \dd y
$$
 must be found so that each element $X_i$ of a basis $\{X_1,\ldots,X_l\}$ of the Vessiot--Guldberg Lie algebra under study becomes Hamiltonian (see \cite{BBHLS} for details). In such a case, we say that $\omega$ is {\em compatible with the Vessiot--Guldberg Lie algebra}. Note that
each $X_i$ is a Hamiltonian vector field with respect to $\omega$ whenever   the Lie derivative of $\omega$ relative to $X_i$ vanishes, that is, $\mathcal{L}_{X_i}\omega=0$. If $\omega$ exists, then all $X_i$ become Hamiltonian vector fields and their
 corresponding Hamiltonian functions $h_i$ are obtained by using the relation $\iota_{X_i}\omega={\rm d}h_i$. 
 The symplectic form $\omega$ induces a Poisson bracket on $C^\infty(\mathbb{R}^2)$ of the form 
 \begin{equation}\label{LB}
 \{\cdot,\cdot\}_\omega\ :\ C^\infty\left(\mathbb{R}^2\right)\times C^\infty\left(\mathbb{R}^2\right)\ni (f,g)\mapsto X_gf\in C^\infty\left(\mathbb{R}^2\right),
 \end{equation}
 with $X_g$ being the Hamiltonian vector field associated to the function $g$. In this way, the functions $ h_1,\ldots,h_l$ and their successive Lie brackets with respect to (\ref{LB}) span a finite-dimensional Lie algebra of functions  that we call a
    {\it  LH algebra} of $V$.

 It has been recently proven that the $8+20$ classes of the GKO classification lead to $4+8$ classes of finite-dimensional Lie algebras of Hamiltonian vector fields~\cite{BBHLS}.
  The final result  is summarized in table~\ref{table1},  where we detail the Lie algebra $ \mathfrak{g}$ isomorphic to the Vessiot--Guldberg Lie algebra spanned by the vector fields $X_i$, an associated symplectic form $\omega$ and the corresponding Hamiltonian functions $h_i$.


\begin{table}[t] {\footnotesize
  \noindent
\caption{{\small The   classification of the 12 finite-dimensional real  Lie algebras of Hamiltonian vector fields on $\mathbb{R}^2$. Note that $\mathfrak{g}=\mathfrak{g}_1\ltimes \mathfrak{g}_2$ means that $\mathfrak{g}$ is the direct sum (as linear subspaces) of $\mathfrak{g}_1$ and $\mathfrak{g}_2$, with $\mathfrak{g}_2$ being an ideal of $\mathfrak{g}$. For I$_{12}$, I$_{14A}$ and I$_{16}$, we have $j=1,\dots,r$ while in  I$_{14B}$ the index $j=2,\dots, r$   since $\eta_1(x)\equiv 1$. In all cases $r\geq 1$. The elements of the basis of the Lie algebra P$_3$ also admit Hamiltonian functions $\bar h_1=h_1+1/4,h_2,h_3$, respectively, spanning a LH algebra isomorphic to $\mathfrak{so}(3)$.
  }}
\label{table1}
\medskip
\noindent\hfill
\begin{tabular}{ l   l l l l }
\hline
 &&&  &\\[-1.5ex]
\#&$\mathfrak{g}$ &Basis of vector fields $X_i$  & Hamiltonian functions $h_i$& $\omega$ \\[+1.0ex]
\hline
 &  & &  &\\[-1.5ex]
P$_1$& $  {\mathfrak{iso}}(2)$& $  { {\partial_x} , \   {\partial_y} , \    y\partial_x - x\partial_y}$ & ${y, \ -x, \ \tfrac 12 (x^2+y^2)},\ 1$ & ${\rm d}x\wedge {\rm d}y$ \\[+1.0ex]
P$_2$& $\mathfrak{sl}(2 )$& $ {\partial_x}, \  {x\partial_x  +  y\partial_y}, \  (x^2  -  y^2)\partial_x  +  2xy\partial_y$ & $\displaystyle{- \frac 1y, \ -\frac xy, \  -\frac{x^2+y^2}{y}}$ & $\displaystyle{    \frac{ {\rm d}x\wedge {\rm  d}y }{ y^{2}  } }$ \\[+2.5ex]
P$_3$& $\mathfrak{so}(3)$ & $ y\partial_x  -  x\partial _y,\ (1  +  x^2  -  y^2)\partial_x  +  2xy\partial_y,$ &$\displaystyle { \frac{-1}{2 (1+x^2+y^2) },\  \frac{ y}{1+x^2+y^2},  }$ &$\displaystyle \frac{\dd x\wedge \dd y}{(1+x^2+y^2)^{2}}$\\[+2.5ex]
& & $2xy\partial_x  +  (1  +  y^2  -  x^2)\partial_y$ &$\displaystyle { - \frac{x}{1+x^2+y^2} }$, \ 1 &  \\[+2.5ex]
 P$_5$& $\mathfrak{sl}(2 )\! \ltimes\! \mathbb{R}^2$&$ {\partial_x},  \ {\partial_y},\  x\partial_x - y\partial_y, \  y\partial_x, \  x\partial_y$  & ${y,\  -x, \ xy,\  \frac 12 y^2,\  -\frac 12 x^2, \ 1}$&$\dd x\wedge \dd y$     \\[+2ex]

I$_1$& $\mathbb{R}$  & $ \partial_x$&$\int^y{f(y')\dd y'}$ & $f(y)\dd x\wedge \dd y$\\[+1.5ex]
I$_4$& $\mathfrak{sl}(2 )$&${  {\partial_x  +  \partial_y}, \   {x\partial _x + y\partial_y}, \  x^2\partial_x  +  y^2\partial_y}$   & $\displaystyle{ \frac 1 {{x-y}   } ,\ \frac{x+y}{2(x-y)},\ \frac{xy}{x-y} }$ &$\displaystyle   {\frac{\dd x\wedge \dd y} {{(x-y)^{2}}}}$ \\[+2.5ex]
I$_5$& $\mathfrak{sl}(2 )$ &${ {\partial_x}, \   { x\partial_x +\frac 12 y\partial_y}, \  x ^2\partial_x  +  xy\partial_y}$ &$\displaystyle { {-\frac{1}{2y^{2}},\ -\frac{x}{2 y^{2}},\ -\frac{x^2}{2 y^{2}  }  }} $  &$\displaystyle{\frac{\dd x\wedge \dd y}{y^3}}$ \\[+2.5ex]
I$_8$& $  { {\mathfrak{iso}}}(1,1)$& ${  {\partial_x},  \  {\partial_y}, \  x\partial_x  - y\partial_y},\ \   $ &${y,\ -x,\ xy,\ 1 }$ & $\dd x\wedge \dd y$ \\[+2.0ex]
I$_{12}$& $\mathbb{R}^{r+1}$& $ {\partial_y} , \  \xi_1(x)\partial_y, \ldots , \xi_r(x)\partial_y$ &$-\int^x \!\! {f(x')\dd x'}, - \int^x \!\! {f(x')\xi_j(x')\dd x'}$ & $f(x)\dd x\wedge \dd y$ \\[+2.0ex]
I$_{14A}$& $\mathbb{R} \ltimes \mathbb{R}^{r}$& ${ {\partial_x}, \   {\eta_1(x)\partial_y} ,\ldots ,\eta_r(x)\partial_y}$  &$y,\  - \int^x {\eta_j(x')\dd x'} $,\quad $1\notin  \langle \eta_j \rangle$ & $ \dd x\wedge \dd y$ \\[+2.0ex]
I$_{14B}$& $\mathbb{R} \ltimes \mathbb{R}^{r}$& ${ {\partial_x}, \   {\partial_y} , \  {\eta_2(x)\partial_y},\ldots ,\eta_r(x)\partial_y}$  &$y,\ -x, \ - \int^x{\eta_j(x')\dd x'},\ 1 $ & $ \dd x\wedge \dd y$ \\[+1.5ex]
I$_{16}$& $ {\mathfrak{h}_2 \! \ltimes \! \mathbb{R}^{r+1}}$ & ${  {\partial_x}, \    {\partial_y}  , \   x\partial_x  - y\partial y,  \  x\partial_y, \ldots, x^r\partial_y}$ & $\displaystyle{ {y,\  -x, \ xy,  \  -\frac{x^{j+1}}{j+1} ,\ 1} }$ & $\dd x\wedge \dd y$ \\[+2.0ex]
    \hline
\end{tabular}
\hfill}
\end{table}


We remark that  in some cases the functions $h_1,\ldots, h_l$ do not span by themselves a Lie algebra of  Hamiltonian functions and a central generator $h_0=1$ must be added in such a manner that $\langle h_1,\ldots,h_l,h_0\rangle$ form a central extension of the initial Vessiot--Guldberg Lie algebra. For instance, the case ${\rm P}_1$ from table~\ref{table1} corresponds to the two-dimensional Euclidean algebra
$ {\mathfrak{iso}}(2)\simeq\langle X_1,X_2,X_3\rangle$, but the Hamiltonian functions $ h_1,h_2,h_3,h_0=1$
span the centrally extended  Euclidean algebra $\overline {\mathfrak{iso}}(2)$ (as in (\ref{za})).  A similar fact arises in classes ${\rm P}_3\simeq \mathfrak{so}(3)$, ${\rm P}_5\simeq \mathfrak{sl}(2 ) \ltimes\mathbb{R}^2$,
${\rm I}_8\simeq  { {\mathfrak{iso}}}(1,1)$ (the $(1+1)$-dimensional Poincar\'e algebra), ${\rm I}_{14B}\simeq \mathbb{R} \ltimes \mathbb{R}^{r}$ and ${\rm I}_{16}\simeq {\mathfrak{h}_2 \! \ltimes \! \mathbb{R}^{r+1}}$. Among them, only the family ${\rm P}_3\simeq \mathfrak{so}(3)$ is a simple Lie algebra so   that $h_0=1$ gives rise to a trivial central extension, namely the LH  algebra is 
  $\mathfrak{so}(3)\oplus\mathbb{R}$; otherwise the central extension is a non-trivial one and it cannot be `removed' (see \cite{BBHLS} for details).

  In this respect, notice    that the appearance of a non-trivial central extension is the  difference between  the families  ${\rm I}_{14B}$ and  ${\rm I}_{14A}$.   We also recall that the LH  algebra corresponding  to the class P$_5$, that is  $\overline{\mathfrak{sl}(2 )\ltimes \mathbb{R}^2}$,  is isomorphic to the two-photon Lie algebra   $\mathfrak{h}_6$ (see~\cite{BBF09,Gilmore} and references therein)
  and, therefore,    to   the  $(1+1)$-dimensional centrally extended Schr\"odinger  Lie algebra~\cite{Schrod}.

We stress that   the Lie algebra  $\mathfrak{sl}(2 )$ appears three times (classes ${\rm P}_2$,  ${\rm I}_4$ and ${\rm I}_5$) which means that there are
{\em different} LH systems sharing isomorphic Vessiot--Guldberg Lie algebras that are non-diffeomorphic, that is, there exists no diffeomorphism mapping the elements of one into the other. In other words, only LH  systems  belonging to each class can be related through a $t$-independent change of variables.   We shall explicitly apply this property throughout the paper. In section 4   we   develop new criteria to easily determine to which class is diffeomorphic a Vessiot--Guldberg Lie algebra isomorphic to $\mathfrak{sl}(2)$ on the plane.


\sect{Determination of non-simple LH systems}

The standard approach for determining a symplectic form $\omega$ turning the elements of a Vessiot--Guldberg Lie algebra $V$ into local Hamiltonian vector fields consists in solving the family of PDEs in $\omega$ of the form $\mathcal{L}_{X_i}\omega=0$ with $X_i$ being any element of $V$~\cite{BBHLS}. Meanwhile, we here show how we can derive $\omega$  out of the Lie algebra and geometric structure of $V$ for non-simple planar Vessiot--Guldberg Lie algebras not diffeomorphic  either to the trivial Lie algebra  I$_{1}$ or  to the  Abelian one I$_{12}$.

Given an $m$-dimensional manifold $M$, a {\it multivector field} on $M$ is an element of the $C^\infty(M)$-module $\mathfrak{X}^\bullet M$ of totally antisymmetric contravariant tensor fields on $M$ of any order. Totally $k$-contravariant multivector fields are called  {\it $k$-multivector fields} and, when $k=2$, {\it bivector fields}.  We write $\mathfrak{X}^kM$ for the $C^\infty(M)$-module of $k$-multivector fields, $\mathfrak{X}^0M$ stands for $C^\infty(M)$ and we fix $\mathfrak{X}^kM=\{0\}$ for $k>\dim M$ and $k<0$. 
The space $\mathfrak{X}^\bullet M$ becomes a $\mathbb{Z}$-graded algebra with respect to the decomposition $\mathfrak{X}^\bullet M=\bigoplus_{i\in \mathbb{Z}}\mathfrak{X}^iM$ when endowed with the $C^\infty(M)$-bilinear exterior product $\wedge:\mathfrak{X}^\bullet M\times \mathfrak{X}^\bullet M\rightarrow \mathfrak{X}^\bullet M$ satisfying 
$$
(P\wedge Q)(\theta_1,\ldots,\theta_{p+q}):= \sum_{\sigma\in S_{p+q}}  (-1)^{{\rm sign}(\sigma)}P(\theta_{\sigma(1)},\ldots,\theta_{\sigma(p)})Q(\theta_{\sigma(p+1)},\ldots,\theta_{\sigma(p+q)}),
$$
with  $S_{p+q}$ being the permutation group of $p+q$ elements, $P\in \mathfrak{X}^pM$, $Q\in \mathfrak{X}^qM$, and $\theta_1,\ldots,\theta_{p+q}$ being arbitrary one-forms on $M$ \cite{IV,CMM97}. 

The natural Lie algebra structure on the space $\mathfrak{X}^1M$ of vector fields on $M$ can be extended to an $\mathbb{R}$-bilinear operation $[\,\cdot\,,\cdot\,]_{\rm SN}:\mathfrak{X}^\bullet M\times \mathfrak{X}^\bullet M\rightarrow \mathfrak{X}^\bullet M$ by requiring $\mathfrak{X}^\bullet M$ to become a graded Lie algebra relative to the decomposition $\mathfrak{X}^\bullet M=\bigoplus_{i\in \mathbb{Z}}\mathfrak{X}^iM$ and considering each element of $\mathfrak{X}^iM$ to have degree $i-1$. The resulting $\mathbb{R}$-bilinear operation is called the {\it Schouten--Nijenhuis bracket} \cite{IV}. A bivector field $\Lambda$ satisfying that $[\Lambda,\Lambda]_{\rm SN}=0$ is called a {\it Poisson bivector}.

Consider a Lie algebra $\mathfrak{g}$. Let $(T(\mathfrak{g}),\otimes)$ be the tensorial algebra relative to the linear space $\mathfrak{g}$ and let $\mathcal{R}$ be the ideal of $T(\mathfrak{g})$ generated by the elements $[v_1,v_2]-(v_1\otimes v_2-v_2\otimes v_1)$ with $v_1,v_2\in \mathfrak{g}$. We call $U(\mathfrak{g}):= T(\mathfrak{g})/\mathcal{R}$ a {\it universal enveloping Lie algebra} associated to $\mathfrak{g}$. Observe that $\mathfrak{g}$ can naturally be considered as a subspace of $U(\mathfrak{g})$. Since $\mathcal{R}$ is an ideal, the tensorial product of $(T(\mathfrak{g}),\otimes)$ gives rise to an $\mathbb{R}$-bilinear product $\widetilde{\otimes}:U(\mathfrak{g})\times U(\mathfrak{g})\rightarrow U(\mathfrak{g})$ turning $U(\mathfrak{g})$ into an $\mathbb{R}$-algebra $(U(\mathfrak{g}),\widetilde{\otimes})$. The Lie algebra structure of $\mathfrak{g}$ can be extended to $U(\mathfrak{g})$ by requiring the extension to become a derivation on each entry with respect to the product  in  $(U(\mathfrak{g}),\widetilde{\otimes})$. This makes $U(\mathfrak{g})$ into a Lie algebra $(U(\mathfrak{g}),[\cdot,\cdot]_{U(\mathfrak{g})})$ \cite{CL99}. A {\it Casimir element} is an element of $U(\mathfrak{g})$ commuting with every element of $\mathfrak{g}$, namely an element in the center of $(U(\mathfrak{g}),[\cdot,\cdot]_{U(\mathfrak{g})})$.  

Repeating the same process as above for $\mathcal{R}$ being the ideal of $T(\mathfrak{g})$ generated by the elements $v_1\otimes v_2-v_2\otimes v_1$ we obtain the {\em symmetric algebra} $S(\mathfrak{g})$ of $\mathfrak{g}$. The Lie algebra structure of $\mathfrak{g}$ can be extended to $S(\mathfrak{g})$ turning this space into a Poisson algebra. As before, $\mathfrak{g}$ is naturally embedded within $S(\mathfrak{g})$. Given a Lie algebra monomorphism $\rho:\mathfrak{g}\rightarrow \mathfrak{X}^1M$ mapping a basis $\{v_{1},\ldots,v_l\}$ of $\mathfrak{g}$ into the vector fields $X_1,\ldots, X_l$, respectively, the linear morphism $\lambda:S(\mathfrak{g})\rightarrow U(\mathfrak{g})$ mapping $\lambda(v_{i_1}\cdot\ldots\cdot v_{i_l})=1/l!\sum_{\sigma \in S_{l}}X_{\sigma(i_1)}\widetilde{\otimes}\ldots \widetilde{\otimes} X_{\sigma(i_l)}$ is a linear isomorphism \cite{BCHLS13Ham,Va84}. This morphism, the {\it symmetryzer morphism}, allows us to interpret the elements of $U(\mathfrak{g})$ as symmetric tensorial elements of $T(\mathfrak{g})$. Moreover, the above mentioned linear morphism  maps the center of $S(\mathfrak{g})$ into the center of $U(\mathfrak{g})$. Hence, if an element of $\mathfrak{g}$ commutes with the whole $S(\mathfrak{g})$ according to the Lie algebra structure of $S(\mathfrak{g})$, then it also commutes with the whole $U(\mathfrak{g})$ respect to its corresponding Lie algebra structure.

Let $V$ be a Vessiot--Guldberg Lie algebra of vector fields on $M$. We define the linear subspace $V\wedge V:= \{ X\wedge Y\,|\,X,Y\in V\}\subset \mathfrak{X}^2M$. Since $\mathcal{L}_X(T_1\otimes T_2)=\mathcal{L}_XT_1\otimes T_2+T_1\otimes \mathcal{L}_XT_2$ for every $X\in V$, $T_1,T_2\in \mathfrak{X}^\bullet M$, then the Lie derivative is a derivation relative to the exterior product of multivector fields and we can define a Lie algebra representation $\varphi_V:X\in V\mapsto L_X\in {\rm End}(V\wedge V)$ for $L_X: \Lambda\in V\wedge V\mapsto \mathcal{L}_X\Lambda \in V\wedge V$. Moreover, due to the graded Lie algebra structure of $\mathfrak{X}^\bullet M$, the Schouten--Nijenhuis bracket of two bivector fields is a 3-vector field. Since 3-vector fields on $\mathbb{R}^2$ vanish everywhere, then the Schouten--Nijenhuis bracket of two bivector fields on   $\mathbb{R}^2$ is zero and all elements of $V\wedge V$ are Poisson bivectors. 

\begin{proposition}\label{Con1} Let $V$ be a Vessiot--Guldberg Lie algebra of planar vector fields. The vector fields of $V$ are Hamiltonian with respect to a bivector field   $\Lambda\in V\wedge V\backslash\{0\}$ if and only if $V$ admits a one-dimensional trivial Lie algebra representation within $V\wedge V$.
\end{proposition}
\begin{proof} 
 If $V$ is a Vessiot--Guldberg Lie algebra of Hamiltonian planar vector fields with respect to a bivector field $\Lambda\in V\wedge V\backslash \{0\}$, then $\mathcal{L}_X\Lambda=0$ for every $X\in V$ and $\Lambda$ spans a one-dimensional trivial Lie algebra representation of $V$.
 Conversely, if $V$ acts trivially and irreducibly via $\varphi_V$ on a one-dimensional vector subspace $W\subset V\wedge V$, then the Lie derivatives of every $\Lambda\in W\backslash\{0\}$ with respect to the elements of $V$ vanish. As $V$ consists of planar vector fields by assumption and $\Lambda$ is a non-zero bivector field, then  the vector fields of $V$ are Hamiltonian relative to $\Lambda$.
\end{proof}

\begin{note} If $\Lambda$ is a zero planar bivector field, then $\mathcal{L}_X\Lambda=0$ for every vector field $X\neq 0$ and $\Lambda$ is a Poisson bivector. Nevertheless, the only Hamiltonian vector fields for a zero Poisson bivector are the zero vector fields. So, $X$ is not a  Hamiltonian vector field relative to $\Lambda$. That is why $\mathcal{L}_X\Lambda=0$ implies that $X$ is Hamiltonian provided $\Lambda\neq 0$. 
\end{note}

\begin{note} The existence of trivial one-dimensional representations of $V$ within $V\wedge V$ depends on the Lie algebra structure of $V$ and also on its geometric structure as a Lie algebra of vector fields. Indeed, isomorphic Vessiot--Guldberg Lie algebras may admit a different number of trivial representations in $V\wedge V$. For instance, $V:= \langle \partial_x,\partial_y\rangle\simeq \mathbb{R}^2 $ gives rise to a unique one-dimensional representation $\langle \partial_x\wedge\partial_y\rangle =V\wedge V$,  while $V:=\langle \partial_x,y\partial_ x\rangle\simeq \mathbb{R}^2$ does not give rise to any such a representation because $V\wedge V=\{0\}$.  
\end{note}

\begin{example} Consider the Lie algebra $V:={\rm I}_{14B}$. In view of table~\ref{table1}, the Lie derivatives of $\Lambda:=\partial_x\wedge \partial_y\in V\wedge V$ with respect to every element of ${\rm I}_{14B}$ vanish. Indeed, 
$$
\left[\partial_x,\partial_x\wedge \partial_y\right]_{\rm SN}=\left[\partial_y,\partial_x\wedge \partial_y\right]_{\rm SN}=0,\quad\     [\eta_j(x)\partial_y,\partial_x\wedge \partial_y]_{\rm SN}=-\frac{\partial \eta_j}{\partial x}\partial_y\wedge\partial_y=0,\quad j=2,\ldots,r.
$$ 
This turns $W=\langle \partial_x\wedge\partial_y\rangle$ into a trivial one-dimensional representation of $V$. Proposition~\ref{Con1} entails that $V$ consists of Hamiltonian vector fields relative to $\Lambda$. This retrieves the result of table~\ref{table1}, where we find the symplectic structure, ${\rm d}x\wedge {\rm d}y$, induced by $\Lambda$. 
Moreover, every Lie algebra ${\rm I}_{14A}$ can be extended to a Lie algebra ${\rm I}_{14B}$ by adding the vector field $\partial_y$. Applying the above procedure, we obtain the same canonical symplectic structure turning the elements of $V$ into Hamiltonian vector fields.
\end{example}

\begin{example} Let us turn to the Lie algebra $V:=I_{16}$. From table~\ref{table1}, we have 
$$
\begin{gathered}
\left[\partial_x,\partial_x\wedge \partial_y\right]_{\rm SN}=0,\qquad [x\partial_x-y\partial_y,\partial_x\wedge \partial_y]_{\rm SN}=-\partial_x\wedge \partial_y+\partial_x\wedge \partial_y=0,\\
\left[\partial_y,\partial_x\wedge \partial_y\right]_{\rm SN}=0,\qquad[x^j\partial_y,\partial_x\wedge\partial_y]_{\rm SN}=-jx^{j-1}\partial_y\wedge\partial_y=0,\qquad j=1,\ldots,r.
\end{gathered}
$$
Hence, $W=\langle \partial_x\wedge \partial_y\rangle \subset  V\wedge V$ is a trivial one-dimensional representation of $V$. As a consequence, $\partial_x\wedge\partial_y$ turns all the elements of $W$ into Hamiltonian vector fields. This is equivalent to the result given in table~\ref{table1}.
\end{example}

The question now is how to determine one-dimensional trivial representations of $V$ in $V\wedge V$. We next devise a method to obtain them for most of the non-simple Lie algebras of table~\ref{table1}.

\begin{theorem}\label{Con2} If $V$ is a planar Vessiot--Guldberg Lie algebra admitting a two-dimensional ideal $I$ such that $I\wedge I\neq \{0\}$ and the elements of $V$ act on $I$ by traceless operators, namely the mappings $\vartheta_X:Y\in I\mapsto [X,Y]\in I$ are traceless for each $X\in V$, then $V$ becomes a Lie algebra of Hamiltonian vector fields with respect to every element of $ I\wedge I\backslash\{0\}$.
\end{theorem}
\begin{proof} If $I\wedge I\neq \{0\}$, the two-dimensional ideal $I=\langle Y_1,Y_2\rangle$ gives rise to a one-dimensional space $I\wedge I$. Since $I$ is an ideal, the Lie brackets of elements   $X\in V$ with elements of $I$ belong to $I$. This ensures that the mappings $\vartheta_X$ are well defined: the Lie derivative with respect to every element of $V$ of an element of $I\wedge I$ belongs to $I\wedge I$. As $I\wedge I\neq\{0\}$ and $V$ consists of planar vector fields, then $Y_1\wedge Y_2\neq 0$ and we can define the dual one-forms  $\{\theta_1,\theta_2\}$ to $\{Y_1,Y_2\}$. Moreover, $\theta_1\wedge \theta_2$ is a volume form on $\mathbb{R}^2$ and we have that 
$$
\begin{aligned}
\mathcal{L}_X (Y_1\wedge Y_2)&=(\mathcal{L}_X Y_1)\wedge Y_2+Y_1\wedge(\mathcal{L}_X Y_2)\\
&=[(\theta_1\wedge\theta_2)(\vartheta_X Y_1,Y_2)+(\theta_1\wedge\theta_2)(Y_1,\vartheta_X Y_2)]Y_1\wedge Y_2\\&={\rm Tr}(\vartheta_X)Y_1\wedge Y_2\\&=0.
\end{aligned}
$$
Since $\Lambda:= Y_1\wedge Y_2\neq 0$ is a planar bivector field, then $X$ becomes a Hamiltonian vector field relative to $\Lambda$ (or to any non-zero bivector field of $V\wedge V$).
\end{proof}

\begin{note} Observe that the condition $I\wedge I\neq \{0\}$ is unavoidable so as to ensure that $V$ admits a compatible Poisson bivector within $V\wedge V\backslash\{0\}$. For instance, the Lie algebra I$_{19}$ with $r=1$ of vector fields of the GKO classification (see \cite{GKP92,BBHLS}) takes the form I$_{19}=\langle \partial_x,\partial_y,x\partial_y,2x\partial_x+y\partial_y,x^2\partial_x+xy\partial_y\rangle$. Note that $I:= \langle\partial_y,x\partial_y\rangle\simeq \mathbb{R}^2$ is an ideal of I$_{19}$ and the elements of I$_{19}$ act on $I$ as traceless operators. Hence, I$_{19}$ satisfies all conditions of theorem \ref{Con2} apart from the fact that $I\wedge I=\{0\}$. In view of table~\ref{table1}, this Lie algebra does not admit a compatible symplectic structure. Hence, the lack of condition $I\wedge I\neq \{0\}$ makes theorem \ref{Con2} to be false.
\end{note}

Let us now show how the above statement allows us to determine many of the Poisson bivector structures appearing in the table~\ref{table1}. 

\begin{example}The Lie algebra P$_1$ admits a two-dimensional ideal $I=\langle \partial_x,\partial_y\rangle$ satisfying that $I\wedge I\neq \{0\}$. Moreover, $\partial_x,\partial_y,y\partial_x-x\partial_y$ act as traceless operators on $I$. In view of theorem \ref{Con2}, the basis $\partial_x\wedge \partial_y$ of $I\wedge I$ becomes a Poisson bivector turning P$_1$ into a Lie algebra of Hamiltonian vector fields. Observe that this Poisson bivector gives rise to a symplectic form $\omega={\rm d}x\wedge {\rm d}y$, which is the one described in table~\ref{table1}. 
\end{example}

\begin{example}
Let us turn to the Lie algebra P$_5$. We have that $I=\langle \partial_x,\partial_y\rangle$ is an ideal of P$_5$ with $I\wedge I\neq 0$ and it is straightforward to prove that all elements of P$_5$ act as traceless linear operators on $I$. Hence, theorem \ref{Con2} ensures that P$_5$ is a Lie algebra of Hamiltonian vector fields relative to the basis $\Lambda:=\partial_x\wedge\partial_y$ of $I\wedge I$. As in the previous example, the symplectic form associated to $\Lambda$ is the canonical one described in table~\ref{table1}.
\end{example}

\begin{example}
The Lie algebra I$_8$ possesses an ideal $I=\langle \partial_x,\partial_y\rangle$ with $I\wedge I\neq \{0\}$ and all the elements of I$_8$ act on $I$ as   traceless mappings. Hence, theorem \ref{Con2} ensures that I$_8$ is a Lie algebra of Hamiltonian vector fields relative to the basis $\Lambda:= \partial_x\wedge \partial_y$ of $I\wedge I$. 
\end{example}

\begin{example}
Consider the Lie algebra of the class I$^{r=2}_{14B}$ given by $V:=\langle {  {\partial_x},   {\partial_y}, x\partial_y  }\rangle$. This Lie algebra possesses an ideal $I=\langle \partial_x,\partial_y\rangle$ satisfying that $I\wedge I\neq \{0\}$. Moreover, the elements of $V$ act on $I$ as traceless operators. Hence, theorem \ref{Con2} ensures that this Lie algebra again consists of Hamiltonian vector fields relative to $\Lambda:=\partial_x\wedge \partial_y$. 
\end{example}

As a practical application, let us apply theorem \ref{Con2} to the subalgebra $V$ appearing in the study of Bernoulli differential equations (\ref{PlanBer}) with $a^R_1(t)=0$. Observe that $V=\langle X_1,X_2,X_3\rangle$, where $X_1,X_2,X_3$ are given in (\ref{VectComBer}),  admits an ideal $I=\langle X_2,X_3\rangle$.  We also have $X_2\wedge X_3\neq0$ and the elements of $V$ act on $I$ as traceless operators. So, $V$ satisfies the conditions given in theorem \ref{Con2} and there exists a non-zero Poisson bivector $\Lambda:=X_2\wedge X_3 \in V\wedge V$ given by (\ref{BP}) turning the elements of $V$ into Hamiltonian vector fields. Thus, Bernoulli differential equations (\ref{PlanBer}) with $a^R_1(t)=0$ are LH systems and a compatible symplectic structure can be derived out of $V$.

Note that the Lie algebra structure of Vessiot--Guldberg Lie algebras does not characterize by itself the existence of a compatible symplectic structure. This is evident from the GKO classification~\cite{GKP92,BBHLS}, where isomorphic Vessiot--Guldberg Lie algebras admit different families of compatible symplectic structures depending on whether they are diffeomorphic or not.


\sect{Determination of $\mathfrak{sl}(2)$-Lie systems}

As already commented, $\mathfrak{sl}(2)$-Lie systems on the plane can belong to one of  the  four non-diffeomorphic classes P$_2$, I$_3$, I$_4$ and I$_5$. 
 This is related to  the fact that the Lie group actions induced by integrating such Lie algebras are not equivariant under a diffeomorphism on $\mathbb{R}^2$. 

We now provide a method to determine to which class of the GKO classification belongs a Lie algebra $V$ of planar vector fields isomorphic to $\mathfrak{sl}(2)$. When the vector fields of $V$ additionally generate a distribution of rank one, $V$ must be diffeomorphic to I$_3$ since this is the only Lie algebra, among I$_3$, I$_4$, P$_2$ and I$_5$, satisfying this property (cf. \cite{GKP92}). The problem to be solved is therefore to classify $V$ when its vector fields generate a distribution of rank two.

Although 
in~\cite{BBHLS} several diffeomorphisms among some $\mathfrak{sl}(2)$-LH systems on the plane, and their corresponding Vessiot--Guldberg Lie algebras, were explicitly determined, we here provide a new  easily verifiable criterium to ensure to which specific class a planar Vessiot--Guldberg Lie algebra isomorphic to $\mathfrak{sl}(2)$ is diffeomorphic to. This is done with no need of deriving the specific diffeomorphism.

 \begin{lemma}\label{lem1} Let $V$ be a Vessiot--Guldberg Lie algebra isomorphic to $\mathfrak{sl}(2)$. We define $S_2(V)$ to  be the space of $2$-contravariant tensor fields spanned by the linear combinations of the tensor fields in the form  $X\otimes Y+Y\otimes X$, with $X,Y\in V$, and we write $S_2(\mathfrak{sl}(2))$ for the space of polynomials of second order in $S(\mathfrak{sl}(2))$. Choose basis $\{v_1,v_2,v_3\}$ and $\{X_1,X_2,X_3\}$ of $\mathfrak{sl}(2)$ and  $V$, respectively, with the same structure constants. This  gives rise to an isomorphism $T:\mathfrak{sl}(2)\simeq V$ and a morphism of $\mathfrak{sl}(2)$-modules
$$
\Psi:S_2( \mathfrak{sl}(2) )\ni \sum_{1=i\le j}^3c_{ij}v_iv_j\ \mapsto\  \sum_{1=i\le j}^3c_{ij}(X_i\otimes X_j+X_j\otimes X_i)\in S_2(V),
$$
i.e.,~$\Psi$ is linear and $\Psi\left(\{v,P\}_{S( \mathfrak{sl}(2))}\right)=\mathcal{L}_{T(v)}\Psi(P)$ for every $v\in \mathfrak{sl}(2)$ and $P\in S_2(\mathfrak{sl}(2))$. If $V$ is diffeomorphic to either {\rm P}$_2$, {\rm I}$_4$ or {\rm I}$_5$, then $\Psi$ is an isomorphism. If $V$ is diffeomorphic to {\rm I}$_3$, then $\Psi$ is not an isomorphism.

\end{lemma}
\begin{proof} It is immediate that $\Psi$ is linear by construction. Let us prove that $\Psi$ is a morphism of $\mathfrak{sl}(2)$-modules. Assuming $\{v_i,v_j\}_{S( \mathfrak{sl}(2))}=\sum_{k=1}^3c_{ijk}v_k$, we obtain $[X_i,X_j]=\sum_{k=1}^3c_{ijk}X_k$ and
$$
\begin{aligned}
\Psi\left(\{v_k,v_iv_j\}_{S( \mathfrak{sl}(2)} \right)&=\Psi\left(\{v_k,v_i\}_{S(\mathfrak{sl}(2))}v_j+v_i\{v_k,v_j\}_{S( \mathfrak{sl}(2))} \right)=\Psi\left(\sum_{l=1}^3c_{kil}v_lv_j+\sum_{l=1}^3v_ic_{kjl}v_l\right)
\\&=\sum_{l=1}^3c_{kil}(X_l\otimes X_j+X_j\otimes X_l)+\sum_{l=1}^3c_{kjl}(X_i\otimes X_l+X_l\otimes X_i)
\\&=(\mathcal{L}_{X_k}X_i)\otimes X_j+X_j\otimes(\mathcal{L}_{X_k}X_i)+X_i\otimes (\mathcal{L}_{X_k}X_j)+(\mathcal{L}_{X_k}X_j)\otimes X_i
\\&=\mathcal{L}_{X_k}(X_i\otimes X_j+X_j\otimes X_i)=\mathcal{L}_{T(v_k)}\Psi(v_iv_j),\qquad \forall k,i,j=1,2,3.
\end{aligned}
$$
Using the linearity of $\Psi$ and the bilinearity of the Poisson bracket in $S_2(\mathfrak{sl}(2))$, we obtain that $\Psi$ is an $\mathfrak{sl}(2)$-module morphism. If $V$ is either  P$_2$, I$_4$ or I$_5$ is straightforward to prove the linearly independence over $\mathbb{R}$ of $\Psi(v_1^2),\Psi(v_2^2),\Psi(v_3^2),\Psi(v_1v_2),\Psi(v_1v_3),\Psi(v_2v_3)$. Since $\dim S_2(\mathfrak{sl}(2))=\dim S_2(V)=6$, then $\Psi$ is an isomorphism. This fact does not change under diffeomorphisms and so it applies to any Lie algebra $V$ diffeomorphic to either P$_2$, I$_4$ and I$_5$.

If $V$ is  I$_3=\langle X_1:= \partial_x,X_2:=x\partial_x,X_3:=x^2\partial_x\rangle$, then $S_2(V)$ is spanned by $\{x^\alpha \partial_x\otimes\partial_x:\alpha=0,\ldots,4\}$. In particular,
$$
\Psi( v_1 v_3-v_2^2)=X_1\otimes X_3+X_3\otimes X_1-2X_2\otimes X_2=0
$$ 
and $\Psi$ is not an isomorphism because $\dim\ker \Psi=1$. These facts do not change under diffeomorphisms so they remain true for any $V$ diffeomorphic to I$_3$.
\end{proof}

\begin{lemma}\label{lem2} A Vessiot--Guldberg Lie algebra $V$ diffeomorphic to either {\rm P}$_2$, {\rm I}$_4$ or {\rm I}$_5$  admits an essentially unique, namely up to proportional constant, $R\in S_2(V)\backslash\{0\}$ such that $\mathcal{L}_XR=0$ for every $X\in V$.
\end{lemma}
\begin{proof} It is well known that there exists only one quadratic Casimir (up to a proportional constant) in $U(  \mathfrak{sl}(2) )$. The $ \mathfrak{sl}(2)$-module isomorphism $\lambda:S(\mathfrak{sl}(2))\simeq U(\mathfrak{sl}(2))$ maps the Casimir $\mathcal{C}$ into an element of $S_2(\mathfrak{sl}(2))$. Using the $\mathfrak{sl}(2)$-module isomorphism $S_2(V)\simeq S_2(  \mathfrak{sl}(2))$, we obtain that there exists essentially one element of $S_2(V)\backslash\{0\}$ whose Lie brackets with the elements of $V$ vanish. Thus, there exists essentially a unique, i.e., up to proportional constant, $R\in S_2(V)\backslash\{0\}$ such that $\mathcal{L}_XR=0$ for every $X\in V$.  \end{proof}

\begin{definition} Given a finite-dimensional real Lie algebra of vector fields $V$, we call {\it Casimir tensor field} of $V$ an element $R\in S_2(V)$ such that $\mathcal{L}_XR=0$ for every $X\in V$.
\end{definition}

\begin{theorem}\label{Easy1} Let $V$ be a Vessiot--Guldberg Lie algebra diffeomorphic to either {\rm P}$_2$, {\rm I}$_4$ or {\rm I}$_5$. Let $R$ be a non-zero Casimir tensor field for $V$. Writing $R=\sum_{\alpha,\beta=1}^2R^{\alpha\beta}\partial_\alpha\otimes \partial_\beta$ with $\partial_1=\partial_x$ and $\partial_2=\partial_y$, we define
$$
\iii(V):= {\rm sign} \left({\rm det}(R^{\alpha\beta}(x)) \right),\qquad \forall x\in {\rm dom}\, V,
$$
where ${\rm dom}\,V$ is the set of generic points of $V$.  If $\iii(V)>0$, then $V$ is locally diffeomorphic to {\rm P}$_2$; when $\iii(V)<0$, then $V$ is locally diffeomorphic to {\rm I}$_4$; if $\iii(V)=0$, then $V$ is locally diffeomorphic to {\rm I}$_5$.
\end{theorem}

\begin{proof}
For the Lie algebras of vector fields P$_2$, I$_4$ and I$_5$ and using their corresponding bases $\{X_1,X_2,X_3\}$ detailed in table~\ref{table1}, we get that 
\begin{equation}\label{TenCas}
R=\frac12(X_1\otimes X_3+X_3\otimes X_1)-X_2\otimes X_2
\end{equation}
satisfies that $\mathcal{L}_{X_i}R=0$ for $i=1,2,3$ and $R\neq 0$. In view of lemma \ref{lem2}, this is essentially the only element of $S_2(V)\backslash\{0\}$ satisfying this property. It is a non-zero Casimir tensor field. After a straightforward calculation and using that  ${\rm dom\, P}_2=\mathbb{R}^2_{y\neq 0}$, ${\rm dom\, I}_4=\mathbb{R}^2_{x\neq y}$ and ${\rm dom\, I}_5=\mathbb{R}^2_{y\neq 0}$    (see \cite{BBHLS}), we obtain
$$
\begin{gathered}
R({\rm P}_2)=-y^2[\partial_x\otimes\partial_x+\partial_y\otimes\partial_y]\ \Rightarrow\  \iii({\rm P}_2)=1,\qquad
R({\rm I}_5)=-\frac{y^2}4[\partial_y\otimes\partial_y]\ \Rightarrow\  \iii({\rm I}_5)=0,\\
R({\rm I}_4)=\frac12(x-y)^2[\partial_x\otimes\partial_y+\partial_y\otimes \partial_x]\ \Rightarrow\  \iii({\rm I}_4)=-1.
\end{gathered}
$$
Observe that the value of $\iii(V)$ is independent of the point  $x\in {\rm dom}\, V$ where we evaluate $\det (R_{\alpha\beta}(x))$ for the Lie algebras P$_2$, I$_4$ and I$_5$. Hence, $\iii(V)$ is well defined. 
Moreover, since $\iii(V)$ depends on the $2\times 2$  matrix of coefficients of $R$, which is unique up to a non-zero multiplicative constant, we have that the value of $\iii(V)$ does not depend on the non-null chosen $R$.

 Let us now prove that $\iii(V)$ is invariant under diffeomorphisms and, therefore, two diffeomorphic Lie algebras have the same $\iii(V)$. Given a local diffeomorphism $\phi:\mathbb{R}^2\rightarrow \mathbb{R}^2$, we can write that $\phi_*R=\sum_{\alpha,\beta=1}^2\bar g^{\alpha\beta}\partial_\alpha\otimes \partial_\beta=\sum_{\mu,\nu=1}^2g^{\mu\nu}A_ \mu^\alpha A_\nu^\beta\partial_\alpha\otimes \partial_\beta$, where $A=(A^\lambda_\sigma)$ is the Jacobian matrix of the diffeomorphism in the initial and final basis $\{\partial_x,\partial_y\}$. In consequence, $\det \bar g^{\alpha\beta}=\det A^2\det g^{\alpha\beta}$ and $\iii(V)$ is invariant under diffeomorphisms. 
 
If $V$ is diffeomorphic to one of the Lie algebras P$_2$,  I$_4$  or I$_5$, then the element $R$ for $V$ is essentially unique and it must be mapped via a diffeomorphism onto an $R\neq 0$ corresponding to  P$_2$, I$_4$ or I$_5$. Since $\iii(V)$ is invariant under diffeomorphisms, we obtain that   $\iii(V)$ is the same as the one for the Lie algebra to which is diffeomorphic to.
\end{proof}

 In the next two sections we apply the above results in the study of planar $\mathfrak{sl}(2)$-LH systems and their equivalence via diffeomorphisms.


\sect{Cayley--Klein  Riccati equations}

Let us consider the so-called complex Riccati equations, namely
\begin{equation}\label{Riccati3}
\frac{{\rm d} z}{{\rm d} t}=a_0(t)+a_1(t)z+a_2(t)z^2,\qquad z\in\mathbb{C},
\end{equation}
with arbitrary $t$-dependent real coefficients $a_0(t),a_1(t)$ and $a_2(t)$. These equations have several applications from a mathematical and physical point of view~\cite{Or12,Ju97,FMR10}. In fact, these  can be mapped into a particular type of planar Riccati equation \cite{Eg07,Wi08} and they also appear in the study of dissipative and irreversible systems \cite{Sc12}. We hereafter propose a generalization of  Riccati equations over the so-called split-complex and dual-Study numbers.

Consider the real plane with coordinates $\{u,v\}$ and an `additional' unit $\iota$ such that
 $\ota^2 \in\{ -1,+1,0\}.$  
Next  we define
  $z:=u+\ota v$ for $(u,v)\in\mathbb{R}^2$.
Assuming that $\iota$ commutes with real numbers, we can write
 $z^2=u^2+\ota^2 v^2 +2 \ota u v.
$ In this way, the number $z$ 
 comprises {\em three} possibilities \cite{Yaglom}:

 \begin{itemize}

 \item If $\ota^2=-1$, we are dealing with the usual {\em complex numbers} $\ota:= {\rm i}$ and $z\in\mathbb C$.

  \item Setting $\ota^2=+1$ we obtain the  so-called {\em split-complex numbers} $z\in\mathbb C^\prime$. The additional unit  is usually known as  the {\em double} or {\em Clifford unit}. 
  
  \item Meanwhile, if we assume $\ota^2=0$, then  $z$ is known as a {\em dual} or {\em Study number}, $z\in\mathbb{D}$.
 
 \end{itemize}

 With these ingredients we    shall    call the {\em Cayley--Klein  Riccati equation}~\cite{Liesym}  the generalization of the complex Riccati equation (\ref{Riccati3}) to   $z:= u+\ota v\in \{ \mathbb C,\mathbb C^\prime, \mathbb D\}$ which, for real $t$-dependent coefficients $a_0(t),a_1(t),a_2(t)$,  gives rise to the system
\begin{equation}\label{CKRE2}
\frac{{\rm d} u}{{\rm d} t}=a_0(t)+a_1(t)u+a_2(t)(u^2+\ota^2 v^2),\qquad \frac{{\rm d} v}{{\rm d} t}=a_1(t)v+a_2(t)2uv.
\end{equation}

Let us prove that (\ref{CKRE2}) is a Lie system for every possible value of $\iota$. The system   (\ref{CKRE2}) is associated with the   $t$-dependent vector field  given by
\be
 X=a_0(t)X_1+a_1(t)X_2+a_2(t)X_3,
 \label{qa}
\ee
 where
\begin{equation}
X_1= \frac{\partial}{\partial u},\qquad X_2= u\frac{\partial}{\partial u}+v\frac{\partial}{\partial v} ,\qquad X_3= (u^2+\iota^2v^2)\frac{\partial}{\partial u}+2uv\frac{\partial}{\partial v} ,
\label{vectRiccati2}
\end{equation}
span a Vessiot--Guldberg Lie algebra $V_\iota\simeq \mathfrak{sl}(2)$ with commutation relations given by
\begin{equation}\label{aa}
[X_1,X_2]=X_1,\qquad [X_1,X_3]=2X_2,\qquad [X_2,X_3]=X_3 .
\end{equation}
Consequently, $X$ takes values in a finite-dimensional Lie algebra $V_\iota$ of vector fields and (\ref{CKRE2}) becomes a Lie system. Let us prove that $X$ is also a LH system.

\begin{proposition}\label{CK} The Cayley--Klein Riccati equation with $t$-dependent real coefficients (\ref{CKRE2}) is a LH system for each value of $\iota$. It admits a Vessiot--Guldberg Lie algebra, $V_\iota$, which is locally diffeomorphic around points of ${\rm dom}\,V_\iota$ to {\rm P}$_2$ when $\iota^2=-1$,   to {\rm I}$_4$ when $\iota^2=1$, and  to {\rm I}$_5$ when $\iota^2=0$. 
\end{proposition}

\begin{proof}Since the Vessiot--Guldberg Lie algebra $V_\iota$ for Cayley--Klein Riccati equations is spanned by the vector fields (\ref{vectRiccati2}) and their commutation relations are given by (\ref{aa}), it turns out that a Casimir tensor field for $V_\iota$ is given by (\ref{TenCas}). Substituying (\ref{vectRiccati2}) in (\ref{TenCas}), we obtain that
$$
R=\iota^2v^2\partial_u\otimes\partial_u-v^2\partial_v\otimes\partial_v\ \Rightarrow\  \iii(V_\iota)=-\,{\rm sign}(\iota^2)
$$
 for points in ${\rm dom}\,V_\iota$. The vector fields of (\ref{vectRiccati2}) span a distribution of rank 2, so $V_\iota$ must be diffeomorphic to one of the Lie algebras P$_2$, I$_4,$ I$_5$. From table~\ref{table1} we obtain that all these Lie algebras consist of Hamiltonian vector fields. Thus, all Cayley--Klein Riccati equations with real $t$-dependent coefficients are LH systems.  Finally, we see in view of theorem \ref{Easy1} that $V_\iota$ is locally diffeomorphic to P$_2$ for $\iota^2=-1$,  to I$_4$ for $\iota^2=1$, and   to I$_5$ for $\iota^2=0$. 
\end{proof}

Proposition \ref{CK} allows us to classify the Vessiot--Guldberg Lie algebra of Cayley--Klein equations according to the value of $\iota^2$. 
Notice that complex Riccati equations (\ref{Riccati3}) are just  the    Cayley--Klein Riccati equations for $\iota^2=-1$ and the vector fields (\ref{vectRiccati2})  coincide with the basis of vector fields of P$_2$ given in table~\ref{table1} provided that $\{x:=u,y:=v\}$. This suggests us   to call     split-complex Riccati equations and dual-Study Riccati equations the Cayley--Klein equations for $\iota^2=1$ and $\iota^2=0$, respectively.
Next, we make use of this  result and table~\ref{table1} to determine the associated symplectic structure for these two latter cases.

Consider the case $\iota^2=1$ and define the new variables $\{x,y\}$ given by
\begin{equation}\label{change}\nonumber
\begin{aligned}
&x:=u+v,\qquad y:=u-v,\qquad 
&u=\tfrac 12(x+y),\qquad v=\tfrac 12 (x-y).
\end{aligned}
\end{equation}
In the new coordinate system, the vector fields (\ref{vectRiccati2})  take the form of the basis of I$_4$ given in table~\ref{table1} such that
${\rm dom\, I}_4=\mathbb{R}^2_{x\neq y}= \mathbb{R}^2_{v\neq 0}$.
 Writing the compatible symplectic two-form and the associated  Hamiltonian functions for the basis of I$_4$ given in table~\ref{table1} in the variables $\{u,v\}$, we obtain that $X_1,X_2,X_3$ are Hamiltonian vector fields with Hamiltonian functions $h_1,h_2,h_3$ relative to the symplectic form $\omega$ with
\begin{equation}
\omega=-\frac{{\rm d} u \wedge {\rm d} v}{ 2 v^2}  ,
\qquad  h_1=\frac{1}{2 v} ,\qquad
h_2=  \frac{u}{2 v}  ,\qquad
h_3=\frac{u^2-v^2}{2 v}.
\nonumber
\end{equation}

Assume now $\iota^2=0$ and $v>0$. The case $v<0$ can be studied analogously giving a similar result.
We define new variables $\{x,y\}$ of the form
\begin{equation}\nonumber
\begin{aligned}
&x:=u,\qquad y:=\sqrt {v},\qquad  &u=x ,\qquad v= y^2. 
\end{aligned}
\end{equation}
By writing the vector fields (\ref{vectRiccati2}) in the new variables, we obtain
 the basis of vector fields appearing in the   Lie algebra I$_5\simeq \mathfrak{sl}(2)$ of   table~\ref{table1} with ${\rm dom\, I}_5=\mathbb{R}^2_{y\neq 0} =  \mathbb{R}^2_{v\neq 0}$.  Hence in the   variables $\{u,v\}$, we find  that
  \begin{equation}
\omega=\frac{{\rm d} u \wedge {\rm d} v}{ 2 v^2}  ,\qquad  h_1=-\frac{1}{2 v} ,\qquad
h_2= - \frac{u}{2 v}  ,\qquad
h_3=-\frac{u^2}{2 v}.
\label{dee}\nonumber
\end{equation}
 
We remark that, independently of the value of $\iota^2$,  the Hamiltonian functions $h_1,h_2,h_3$ satisfy (see (\ref{aa}))
\begin{equation*}
\label{sl2Rh}
\{h_1,h_2\}_\omega=-h_1,\qquad \{h_1,h_3\}_\omega=-2h_2,\qquad \{h_2,h_3\}_\omega=-h_3.
\end{equation*}
Hence, $ (\langle h_1,h_2,h_3\rangle,\{\cdot,\cdot\}_\omega)$ is always a LH algebra for the system (\ref{CKRE2})  isomorphic to $\mathfrak{sl}(2)$
and
$$
h=a_0(t)h_1+a_1(t)h_2+a_2(t)h_3
$$
 is a  $t$-dependent Hamiltonian function associated to the $t$-dependent vector field (\ref{qa}) and, therefore, to the system (\ref{CKRE2}).

We also stress that the Cayley--Klein system (\ref{CKRE2})  for $\iota^2\in \{0,1\}$ appears in a method to map diffusion-type equations into a simpler form which can be easily integrated~\cite{SSVG11,SSVG14}.


\sect{Other planar $\mathfrak{sl}(2)$-LH systems}

In this section we present some $\mathfrak{sl}(2)$-LH systems of mathematical and physical  interest, namely,  coupled Riccati, Milne--Pinney (which comprises   the Smorodinsky--Winternitz system and the harmonic oscillator, both  with a $t$-dependent frequency),  second-order Kummer--Schwarz and planar diffusion  equations. 
Furthermore,  we also study, according to table~\ref{table1}, the equivalence among them and the Cayley--Klein Riccati equations introduced in the previous section, that is,   we establish, by applying the results of section 4,  which of all of the above systems are locally diffeomorphic.
To keep notation simple, hereafter we say that a second-order differential equation is a LH system when the first-order system obtained from it by adding a new variable $y:={\rm d}x/{\rm d}t$, is a LH one.


\subsection{Coupled Riccati differential equations}

Consider the system of coupled differential Riccati equations \cite{Mariton}
\begin{equation}
\frac{{\rm d}x}{{\rm d}t}=a_0(t)+a_1(t)x+a_2(t)x^2,\qquad \frac{{\rm d}y}{{\rm d}t}=a_0(t)+a_1(t)y+a_2(t)y^2,
\label{cR}
\end{equation}
which appears as a particular case of systems of Riccati equations studied in~\cite{BCHLS13Ham,CGLS}. This system can be expressed as a $t$-dependent vector field (\ref{qa})
where
\begin{equation}
X_1= \frac{\partial}{\partial x}+ \frac{\partial}{\partial y},  \qquad X_2= x\frac{\partial}{\partial x}+y\frac{\partial}{\partial y} ,\qquad X_3= x^2\frac{\partial}{\partial x}+y^2\frac{\partial}{\partial y} ,
\nonumber
\end{equation}
so that these  vector fields  exactly reproduce those given in   table~\ref{table1}  for the   class I$_4$ which, in turn, means that this system is
locally diffeomorphic to the   split-complex Riccati equation, namely (\ref{CKRE2}) with $\iota^2=1$.


\subsection{Milne--Pinney equations}
The Milne--Pinney equation \cite{LA08,Mi30,Pi50} has the following expression
\begin{equation}\label{mp}
\frac{\dd^2x}{\dd t^2}=-\omega^2(t)x+\frac{c}{x^3},
\end{equation}
where $\omega(t)$ is any $t$-dependent function and $c$ is a real constant. We remark that this system is, in fact,  the one-dimensional Ermakov  system as well as the  Smorodinsky--Winternitz system~\cite{WSUF65}     with a $t$-dependent frequency. The $c$-term can be understood as a  Rosochatius potential (or a centrifugal barrier when $c>0$)   in its Hamiltonian form~\cite{nonlinear}. When $c$ vanishes, the system (\ref{mp}) reduces to the harmonic oscillator with a
$t$-dependent frequency.

Next, by introducing $y:= \dd x/\dd t$, we rewrite  \eqref{mp} as a  system of   first-order  differential equations
\be
\frac{\dd x}{\dd t}=y,\qquad \frac{\dd y}{\dd t}=-\omega^2(t)x+\frac{c}{x^3},
\label{FirstLie}
\ee
which has an  associated  $t$-dependent vector field
$
X=X_3+\omega^2(t)X_1,
$
where
\begin{equation}\label{FirstLieA}  
X_1=-x\frac{\partial}{\partial y},\qquad X_2=\frac 12 \left(y\frac{\partial}{\partial y}-x\frac{\partial}{\partial x}\right),\qquad X_3=y\frac{\partial}{\partial x}+\frac{c}{x^3}\frac{\partial}{\partial y},
\end{equation}
span a   Lie algebra $V^{\rm MP}$ of vector fields isomorphic to $\mathfrak{sl}(2)$ with commutation relations given by (\ref{aa}). It has been proven in~\cite{BBHLS} that the Milne--Pinney equations (\ref{FirstLie}) comprise the three different types of  $\mathfrak{sl}(2)$-LH systems according to the value of the constant $c$ as follows.

\begin{proposition} The system  (\ref{FirstLie})  is a LH system of class    {\rm P}$_2$ for $c>0$; {\rm I}$_4$ for $c<0$; and {\rm I}$_5$ for $c=0$.
\end{proposition}

Since $V^{\rm MP}$ spans a distribution of rank two on the plane and it is isomorphic to $\mathfrak{sl}(2)$, it must be diffeomorphic to P$_2$, I$_4$ or I$_5$. We can therefore recover, straightforwardly, the above proposition as a particular case of theorem \ref{Easy1}. Since the vector fields (\ref{FirstLieA})  satisfy the commutation relations (\ref{aa}), we obtain that (\ref{TenCas}) reads
$$
\begin{aligned}
R =-\frac 14\left[xy(\partial_x\otimes\partial_y+\partial_y\otimes\partial_x)+x^2\partial_x\otimes\partial_x+\left(y^2+\frac{4c}{x^2}\right)\partial_y\otimes\partial_y\right]
\end{aligned}
$$
and $\iii( V^{\rm MP} )=c$. Using this and theorem \ref{Easy1}, we retrieve the same result given in previous proposition. Therefore, like the Cayley--Klein Riccati equations (\ref{CKRE2}), the Milne--Pinney ones include the three possibilities of Vessiot--Guldberg Lie algebras isomorphic to $\mathfrak{sl}(2)$ of Hamiltonian vector fields.


\subsection{Second-order Kummer--Schwarz equation}

This is the second-order differential equation given by
\begin{equation}\label{KS1}\nonumber
 \frac{{\rm d}^2x}{{\rm d}t^2}=\frac{3}{2x}\left(\frac{{\rm d}x}{{\rm d}t}\right)^2-2c\, x^3+2\eta(t)x,
\end{equation}
where $c$ is a real constant and $\eta(t)$ is an arbitrary $t$-dependent function. We define $y:= {\rm d}x/{\rm d}t$
and we obtain a  first-order system which has been studied in~\cite{CGL11}
\begin{equation}\label{KSsys}
\frac{{\rm d}x}{{\rm d}t}=y,\qquad
\frac{{\rm d}y}{{\rm d}t}=\frac{3}{2}\frac{y^2}{x}-2c\, x^3+2\eta(t)x.
\end{equation}
This system has an associated $t$-dependent vector field $X=X_3+\eta(t)X_1,$ where the vector fields
\be
 X_1=2x\frac{\partial}{\partial y},\qquad X_2=x\frac{\partial}{\partial x}+2y\frac{\partial}{\partial y},\qquad X_3=y\frac{\partial}{\partial x}+\left(\frac{3}{2}\frac{y^2}{x}-2c\, x^3\right)\frac{\partial}{\partial y}
\label{KS}
\ee
form a basis of a Lie algebra $V^{\rm KS}$ isomorphic to $\mathfrak{sl}(2)$ with commutation relations (\ref{aa}).

It can be proven that $V^{\rm KS}$   comprises, once more, the three Vessiot--Guldberg Lie algebras of Hamiltonian vector fields isomorphic to $\mathfrak{sl}(2)$ given in  table~\ref{table1}  according to  the value of the parameter $c$~\cite{BBHLS}.

\begin{proposition} The system  (\ref{KSsys}) is a LH system of class  {\rm P}$_2$ for $c>0$; {\rm I}$_4$ for $c<0$; and {\rm I}$_5$ for $c=0$.
\label{propo62}
\end{proposition}

This   statement  was obtained in~\cite{BBHLS} by deriving the explicit diffeomorphisms from the Lie algebra spanned by (\ref{KS}) to one of the Lie algebras given in table~\ref{table1}.  In any case, we can retrieve these results more easily by using theorem \ref{Easy1}. Indeed, as the vector fields (\ref{KS}) span a distribution of rank two, the theorem \ref{Easy1} applies. Moreover, these vector fields satisfy the commutation relations (\ref{aa}) and we find that
$$
R=\frac 12(X_1\otimes X_3+X_3\otimes X_1)\!-\!X_2\otimes X_2=\!-x^2\partial_x\otimes\partial_x\!-\!xy(\partial_x\otimes\partial_y\!+\!\partial_y\otimes\partial_x)-(y^2+4cx^4)\partial_y\otimes\partial_y.
$$
Hence, $\iii(V^{\rm KS})=c$ and, in view of theorem \ref{Easy1}, we recover the results given in   proposition~\ref{propo62}.


\subsection{Planar diffusion Riccati system}

A diffusion equation can be transformed into a simpler PDE by solving a system of seven first-order ordinary differential equations (see \cite{SSVG11} and \cite[p.~104]{SSVG14} for details). This system can be easily solved by integrating its projection onto $\mathbb{R}^2$ given by
\begin{equation}\label{diff}
\frac{{\rm d}x}{{\rm d}t}=-b(t)+2c(t)x+4a(t)x^2+a(t)c_0y^4,\qquad \frac{{\rm d}y}{{\rm d}t}=\bigl(c(t)+4a(t)x\bigr)y,
\end{equation}
where $a(t),b(t)$ and $c(t)$ are arbitrary $t$-dependent functions and $c_0\in \{0,1\}$. We call this system {\it planar diffusion Riccati system}, which is related to the $t$-dependent vector field
$$
X=a(t)X_3-b(t)X_1+c(t)X_2,
$$
where
\begin{equation}\label{PlanarDif}
X_1=\frac{\partial}{\partial x},\qquad X_2=2x\frac{\partial}{\partial x}+y\frac{\partial}{\partial y}, \qquad X_3=(4x^2+c_0y^4)\frac{\partial}{\partial x}+4xy\frac{\partial}{\partial y},
\end{equation}
satisfy the commutation relations
\begin{equation}\label{con}
[X_1,X_2]=2X_1,\qquad [X_1,X_3]=4X_2,\qquad [X_2,X_3]=2X_3.
\end{equation}
Consequently, they span a  Vessiot--Guldberg Lie  algebra $V^{\rm PDR}$  isomorphic to $\mathfrak{sl}(2)$. Let us use again theorem \ref{Easy1} to determine to which one of the Lie algebras of table~\ref{table1} is $V^{\rm PDR}$ diffeomorphic to. As the vector fields (\ref{PlanarDif}) satisfy (\ref{con}), the Casimir tensor field $R$  (\ref{TenCas}) turns out  to be
$$
R=c_0y^4\partial_x\otimes\partial_x-y^2\partial_y\otimes\partial_y\ \Rightarrow\  \iii(V^{\rm PDR})=-c_0.
$$
Since the vector fields $X_1,X_2,X_3$ span a distribution of rank two, the theorem \ref{Easy1} applies. In view of the latter, the system (\ref{diff}) for $c_0=1$ is diffeomorphic to I$_4$ and for $c_0=0$ to I$_5$. Indeed, for $c_0=1$ the change of variables
$$
u:=2x+y^2,\qquad v:=2x-y^2,\qquad x=\tfrac 12( u+v),\qquad y=\sqrt{u-v}
$$
maps (\ref{PlanarDif}) into a basis of I$_4$ whose elements are proportional to those ones given table~\ref{table1} after a relabeling of the variables. Writing the symplectic structure and the Hamiltonian functions given in table \ref{table1} in the initial coordinate system $\{x,y\}$, we obtain
$$
\omega=-\frac{{\rm d}x\wedge {\rm d}y}{y^3},\qquad h_1=\frac{1}{2y^2},\qquad h_2=\frac{x}{y^2},\qquad h_3=2\,\frac{x^2}{y^2}-\frac 12\, y^2,
$$
which satisfy
$$
\{h_1,h_2\}_\omega=-2 h_1,\qquad \{h_1,h_3\}_\omega=-4 h_2,\qquad \{h_2,h_3\}_\omega=-2 h_3.
$$

For the case $c_0=0$, we have that the vector fields (\ref{PlanarDif}) form a  basis of I$_5$ (see table~\ref{table1}). Hence, their associated symplectic form and some corresponding Hamiltonian functions can easily be obtained from table~\ref{table1}. The main result of this section can be summarized as follows.

\begin{proposition} The planar diffusion Riccati system  (\ref{diff}) is a LH system of class  {\rm I}$_4$ for $c_0=1$; and {\rm I}$_5$ for $c_0=0$.
\end{proposition}


\subsection{Equivalence among  planar $\mathfrak{sl}(2)$-LH   systems}

By taking into account the previous results, we are led to the following statement.

\begin{theorem}\label{Main} The $\mathfrak{sl}(2)$-LH systems (\ref{CKRE2}), (\ref{cR}), (\ref{FirstLie}), (\ref{KSsys}) and  (\ref{diff}) are equivalent through local diffeomorphisms whenever they belong to the same class in  table~\ref{table1}, that is,

\begin{itemize}
\item {\rm P}$_2$:  Milne--Pinney  and Kummer--Schwarz   equations for $c>0$ as well as complex Riccati equations with $t$-dependent real coefficients.
\item {\rm I}$_4$:  Milne--Pinney  and Kummer--Schwarz   equations for $c<0$, coupled Riccati equations, split-complex Riccati equations and the planar diffusion Riccati system with $c_0=1$. All of them with $t$-dependent real coefficients.
\item {\rm I}$_5$:  Milne--Pinney  and Kummer--Schwarz   equations for $c=0$ as well as dual-Study  Riccati equations, planar diffusion Riccati systems with $c_0=0$ and the harmonic oscillator with $t$-dependent frequency.
\end{itemize}

\end{theorem}

Only within each class,  these  systems  are locally diffeomorphic  and, therefore,  there   exists a
local $t$-independent change of variables mapping  one into another.  For instance, there does not exist any diffeomorphism on $\mathbb{R}^2$ mapping the Milne--Pinney  and Kummer--Schwarz equations with $c\ne 0$  to the harmonic oscillator with a $t$-dependent frequency as the latter  is a LH system of class  I$_5$ and the previous ones do not. Our results also  explain the existence of the known diffeomorphism mapping   Kummer--Schwarz equations to Milne--Pinney equations provided their constant  $c$ shares the same sign \cite{LA08}.


\sect{Two-photon LH systems}

Let us study two different LH systems that belong to the same class  P$_5$: systems related to $t$-dependent quadratic Hamiltonians and second-order Riccati equations in Hamiltonian form. As a consequence, we shall  prove their equivalence through diffeomorphisms.

The elements of the basis $\{X_1,\dots,X_5\}$ of P$_5$ written in table~\ref{table1} satisfy the Lie brackets
\[
\begin{array}{llll}
[X_1,X_2] = 0, &\quad [X_1,X_3] =X_1,&\quad [X_1,X_4] =0,&\quad [X_1,X_5] =X_2,\\[2pt]
[X_2,X_3] =-X_2,&\quad[X_2,X_4] =X_1,&\quad [X_2,X_5] =0,&\quad
[X_3,X_4] =-2X_4, \\[2pt]
[X_3,X_5] =2X_5,&\quad
[X_4,X_5] =-X_3. &\quad   &
\end{array}
 \]
Hence, they span a Lie algebra isomorphic to $\mathfrak{sl}(2 )\ltimes \mathbb{R}^2$, where $\mathbb{R}^2\simeq \langle X_1,X_2 \rangle$ and $\mathfrak{sl}(2 )\simeq\langle X_3,X_4,X_5 \rangle$.  Observe that this Lie algebra satisfies the conditions given by theorem \ref{Con2}, hence  such vector fields are Hamiltonian relative  to the Poisson bivector $\Lambda=X_1\wedge X_2=\partial_x\wedge\partial_y$ or, equivalently, the associated symplectic structure $\omega=\dd x\wedge \dd y$. This retrieves in an algebraic/geometric manner the result obtained in \cite{BBHLS}.

The corresponding Hamiltonian functions for $X_1,\ldots, X_5$ must be enlarged with a central generator $h_0=1$ giving rise to the centrally extended Lie algebra $\overline{\mathfrak{sl}(2 )\ltimes \mathbb{R}^2}$ which is, in fact, isomorphic to the two-photon Lie algebra $\mathfrak{h}_6=\langle h_1,\dots,h_5,h_0 \rangle$~\cite{BBF09,Gilmore}. That is why we shall call these systems {\it two-photon LH systems}. The commutation relations of this Lie algebra read
\be
\begin{array}{llll}
\{h_1,h_2\}_\omega=h_0, &\quad \{h_1,h_3\}_\omega=-h_1,&\quad \{h_1,h_4\}_\omega=0,&\quad \{h_1,h_5\}_\omega=-h_2,\\[2pt]
\{h_2,h_3\}_\omega=h_2,&\quad\{h_2,h_4\}_\omega=-h_1,&\quad \{h_2,h_5\}_\omega=0,&\quad
\{h_3,h_4\}_\omega=2h_4, \\[2pt]
\{h_3,h_5\}_\omega=-2h_5,&\quad
\{h_4,h_5\}_\omega=h_3, &\quad   \{h_0,\cdot\}_\omega=0 . &
\end{array}
\label{twophoton}
\ee
Notice that  $\mathfrak{h}_6\simeq \mathfrak{sl}(2)
\ltimes \mathfrak{h}_3$, where $\mathfrak{h}_3\simeq  \langle h_0,h_1,h_2\rangle$ is the Heisenberg--Weyl Lie algebra and   $\mathfrak{sl}(2 ) \simeq \langle h_3,h_4,h_5\rangle $. Since  $\mathfrak{h}_4\simeq  \langle h_0,h_1,h_2,h_3\rangle$ is the harmonic oscillator algebra (isomorphic to the LH algebra
  $  {\overline {\mathfrak{iso}}}(1,1)$ of the class I$_8$), we have the inclusions
$\mathfrak{h}_3\subset \mathfrak{h}_4\subset \mathfrak{h}_6$.


\subsection{$t$-dependent quadratic Hamiltonians}

We now study $t$-dependent quadratic Hamiltonians
\begin{equation}\label{disharoscsys}
h(t,q,p)=\alpha(t)\, \frac{p^2}{2}+\beta(t)\, \frac{pq}{2}+\gamma(t)\, \frac{q^2}{2}+\delta(t)p+\epsilon(t)q+\phi(t),
\end{equation}
where $\alpha(t),\beta(t),\gamma(t),\delta(t),\epsilon(t),\phi(t)$ are arbitrary real $t$-dependent functions ~\cite{CR03}. As particular cases, (\ref{disharoscsys}) describes certain damped and/or dissipative harmonic oscillators \cite{UY02}, electric charges in monochromatic electric fields \cite{CR03}, etc. 
The corresponding  Hamilton equations read
\begin{align}\label{disharoscsys2}\nonumber
\frac{{\rm d}q}{{\rm d}t}&=\frac{\partial h}{\partial p}=\alpha(t)\, p+\beta(t)\frac{q}{2}\, +\delta(t),\nonumber\\
\frac{{\rm d}p}{{\rm d}t}&=-\frac{\partial h}{\partial q}=-\left( \beta(t)\, \frac{p}{2}+\gamma(t)q+\epsilon(t)\right) .
\end{align}
This system has an associated $t$-dependent vector field
\begin{equation*}
X=\delta(t)X_1-\epsilon(t)X_2+ \frac{\beta(t)}{2}X_3+\alpha(t)X_4-\gamma(t)X_5,
\end{equation*}
where $X_1,\ldots, X_5$ are, up to a trivial change of variables $x:=q$ and $y:=p$, the vector fields of the basis of P$_5$ given in table~\ref{table1}. Hence, their Hamiltonian functions with respect to the symplectic structure $\omega=\dd q\wedge \dd p$ can easily be obtained from table~\ref{table1}.


\subsection{Second-order Riccati equation}

Second-order Riccati equations, which were  recently studied using the theory of Lie systems in~\cite{CLS12}, read
\begin{equation}\label{NLe}
\frac{\dd^2x}{\dd t^2}+\bigl( f_0(t)+f_1(t)x \bigr) \frac{\dd x}{\dd t}+c_0(t)+c_1(t)x+c_2(t)x^2+c_3(t)x^3=0,
\end{equation}
with
$$
f_1(t)=3\sqrt{c_3(t)},\qquad f_0(t)=\frac{c_2(t)}{\sqrt{c_3(t)}}-\frac{1}{2c_3(t)} \,\frac{\dd c_3(t)}{\dd t}, \qquad c_3(t)> 0,
$$
where $c_0(t),c_1(t),c_2(t)$ are arbitrary $t$-dependent functions  and $c_3(t)$ is a non-negative function. 
This differential equation arises by reducing third-order linear differential equations through a dilation symmetry and a $t$-reparametrization \cite{CRS05}.

The key point is that a quite general family of second-order Riccati equations (\ref{NLe}) admits  a   $t$-dependent Hamiltonian (see \cite{CLS12,CRS05}) given by
\begin{equation*}
 h(t,x,p)= -2\sqrt{-p}- p\left(a_0(t)+a_1(t)x+a_2(t)x^2 \right) ,\qquad p<0,
\end{equation*}
where  $a_0(t),a_1(t),a_2(t)$ are certain functions related to the
  $t$-dependent coefficients of (\ref{NLe}). The corresponding
  Hamilton equations are
\begin{equation}
\begin{aligned}\label{Hamil}
\frac{\dd x}{\dd t}&=\frac{\partial h}{\partial p} =\frac{1}{\sqrt{-p}}-a_0(t)-a_1(t)x-a_2(t)x^2,\\
\frac{\dd p}{\dd t}&=-\frac{\partial h}{\partial x}=  p\left(a_1(t)+2a_2(t)x\right),
\end{aligned}
\end{equation}
and the associated $t$-dependent vector field has the expression
\begin{equation}
X=Y_1-a_0(t)Y_2-a_1(t)Y_3-a_2(t)Y_4, \label{F2}\nonumber
\end{equation}
where
\be
Y_1=\frac{1}{\sqrt{-p}} \, \frac{\partial}{\partial x},\quad\quad
Y_2=\frac{\partial}{\partial x},\quad\quad
Y_3=x\frac{\partial}{\partial x}-p\frac{\partial}{\partial p},\quad\quad
Y_4=x^2\frac{\partial}{\partial x}-2xp\frac{\partial}{\partial p}.
\label{uk}
\ee
Another vector field
\be
Y_5=\frac{x}{\sqrt{-p}}\, \frac{\partial}{\partial x}+2\sqrt{-p} \, \frac{\partial}{\partial p}
\label{ul}
\ee
is required in order to close a Lie algebra $V^{\rm SR}$, whose   non-vanishing commutation relations read
\[
\begin{array}{llll}
 \left[Y_1,Y_3\right]=\frac 12Y_1, &\quad [Y_1,Y_4]=Y_5, &\quad
[Y_2,Y_3]=Y_2,&\quad [Y_2,Y_4]=2Y_3,\\[2pt]
\left[Y_2,Y_5\right]=Y_1,
&\quad \left[Y_3,Y_4\right]=Y_4,&\quad [Y_3,Y_5]=\frac 12 Y_5 .
\end{array}
\]
Hence, $V^{\rm SR}\simeq \mathfrak{sl}(2 )\ltimes \mathbb{R}^2$  where  $\mathbb{R}^2\simeq \langle Y_1,Y_5 \rangle$ and $\mathfrak{sl}(2 )\simeq\langle Y_2,Y_3,Y_4 \rangle$.

Next, observe that  $I=\langle Y_1,Y_5\rangle$ is an ideal of $V^{\rm SR}$ such that $I\wedge I\neq \{0\}$ and that  it can be proven that all elements of $V^{\rm SR}$ act on $I$ as traceless operators. Therefore, in view of theorem \ref{Con2}, this Lie algebra consists of Hamiltonian vector fields with respect to the Poisson bivector $\Lambda:=\tfrac 12 Y_1\wedge Y_5=\partial_x\wedge \partial_p$ or, equivalently,  to the symplectic structure $\omega={\rm d}x\wedge {\rm d}p$.

The   Hamiltonian functions corresponding to the vector fields (\ref{uk}) and (\ref{ul})
 turn out to be
\begin{eqnarray*} 
\tilde h_1=-2\sqrt{-p}, \qquad  \tilde h_2=p,\qquad \tilde h_3=xp,\qquad  \tilde h_4=x^2p,\qquad \tilde h_5=-2x\sqrt{-p},
\end{eqnarray*}
which span along with $\tilde h_0=1$ a   Lie algebra of functions  isomorphic to the two-photon Lie algebra $\mathfrak{h}_6$ with non-vanishing Lie brackets   given by
\be
\begin{array}{llll}
  \{\tilde h_1,\tilde h_3\}_\omega =-\frac  12 \tilde h_1 , &\quad \{\tilde h_1,\tilde h_4\}_\omega =-\tilde h_5,& \quad \{\tilde h_1,\tilde h_5\}_\omega =2\tilde h_0, &\quad
\{\tilde h_2,\tilde h_3\}_\omega =-\tilde h_2, \\
  \{\tilde h_2,\tilde h_4\}_\omega =-2\tilde h_3,& \quad \{\tilde h_2,\tilde h_5\}_\omega =-\tilde h_1, & \quad \{\tilde h_3,\tilde h_4\}_\omega =-\tilde h_4, &\quad
\{\tilde h_3,\tilde h_5\}_\omega =-\frac 12 \tilde h_5  .
\end{array}
\label{ComRel2}\nonumber
\ee
 Indeed, it can be seen that the functions
 $$
 \widehat{h}_1=-\frac{1}{\sqrt{2}}\tilde{h}_5,\qquad \widehat{h}_2=\frac{1}{\sqrt{2}}\tilde{h}_1,\qquad \widehat{h}_3=-2\tilde{h}_3,\qquad \widehat{h}_4=-\tilde{h}_4,\qquad \widehat{h}_5=\tilde{h}_2,\qquad \widehat{h}_0=\tilde{h}_0
 $$
close the same commutation relations that the basis of $\mathfrak{h}_6$ given in (\ref{twophoton}).
The main results of this section are then summarized as follows.

\begin{proposition} The Hamilton equations (\ref{disharoscsys2}) and (\ref{Hamil}) for $t$-dependent quadratic Hamiltonians and second-order Riccati equations (\ref{NLe}), respectively, determine LH systems of class {\rm P}$_5$ with LH algebras isomorphic to the two-photon one $\mathfrak{h}_6$. Consequently all of these systems are locally diffeomorphic.
\end{proposition}


\sect{Projective Schr\"odinger equations on $\mathbb{C}\mathbb{P}^1$}

Let us describe a new application of LH systems on the plane. Consider the Schr\"odinger equations on $\mathbb{C}^2$ given by
\begin{equation}\label{Sch}
{\rm i}\frac{\rm d}{{\rm d}t}\left[\begin{array}{c}
\Psi_1\\
\Psi_2
\end{array}
\right]=\left[\begin{array}{cc}
\lambda_1(t)&\beta(t)\\
\bar{\beta}(t)&\lambda_2(t)
\end{array}
\right]\left[\begin{array}{c}
\Psi_1\\
\Psi_2\\
\end{array}
\right],
\end{equation}
where $\lambda_1(t),\lambda_2(t)$ are arbitrary $t$-dependent real functions, $\beta(t)$ is an arbitrary $t$-dependent complex function and we assume $\hbar=1$ for simplicity. 
Let us construct the projection of this $t$-dependent Schr\"odinger equation onto the projective space $\mathbb{C}_\times^2\backslash \mathbb{C}^\times\simeq \mathbb{C}\mathbb{P}^1$, where $\mathbb{C}^\times:= \mathbb{C}-\{0\}$ and $\mathbb{C}^2_\times=\mathbb{C}^2\backslash\{(0,0)\}$. Observe that $(\Psi_1,\Psi_2),({\Phi}_1,{\Phi}_2)\in \mathbb{C}^2_\times$ belong to the same equivalence class of $\mathbb{C}\mathbb{P}^1$, a so-called {\it ray}, if and only if 
$ \Psi_1{\Phi}_2-{\Phi}_1\Psi_2=0.$ Hence, elements $(\Psi_1,\Psi_2)\in U_1:= \mathbb{C}\times\mathbb{C}^\times$ belonging to the same ray give rise to the same complex number $\Psi_1\Psi_2^{-1}\in\mathbb{C}$.
Using this, we can introduce a well-defined local coordinate system $\pi_1:[(\Psi_1,\Psi_2)]\in U_1\subset \mathbb{C}\mathbb{P}^1\mapsto z:= \Psi_1\Psi_2^{-1}\in\mathbb{C}\simeq\mathbb{R}^2$. Similarly, a second coordinate system can be defined on $U_2:=\mathbb{C}^\times\times \mathbb{C}$.  A simple calculation shows that the projection of (\ref{Sch})  under $\pi_1$ becomes
$$
{\rm i}\frac{{\rm d}z}{{\rm d}t}=\beta(t) -\bar \beta(t)z^2+(\lambda_1(t)-\lambda_2(t))z.
$$
Writing $z=x+ {\rm i}y$ and $\beta(t)=\beta_x(t)+{\rm i}\beta_y(t)$, we obtain
\begin{equation}\label{eqf}
\begin{aligned}
\frac{{\rm d}x}{{\rm d}t}&=-\beta_x(t)2xy+\beta_y(t)(x^2-y^2+1)+(\lambda_1(t)-\lambda_2(t))y,\\
\frac{{\rm d}y}{{\rm d}t}&=\beta_x(t)(x^2-y^2-1)+\beta_y(t)2xy-(\lambda_1(t)-\lambda_2(t))x.
\end{aligned}
\end{equation}
This system of differential equations describes the integral curves of the $t$-dependent vector field $X=-\beta_x(t)X_3+\beta_y(t)X_2+(\lambda_1(t)-\lambda_2(t))X_1$, where
$$
X_1=y\frac{\partial}{\partial x}-x\frac{\partial}{\partial y},\quad X_2=(1+x^2-y^2)\frac{\partial}{\partial x}+2xy\frac{\partial}{\partial y},\quad X_3=2xy\frac{\partial}{\partial x}+(1+y^2-x^2)\frac{\partial}{\partial y}.
$$
These are exactly the vector fields appearing in the Lie algebra P$_3$ of table~\ref{table1}. So, they span a Lie algebra of vector fields isomorphic to $\mathfrak{so}(3)\simeq  \mathfrak{su}(2)$ and $X_1,X_2,X_3$ are Hamiltonian with respect to the symplectic form given in table~\ref{table1}, namely
$$
\omega=\frac{{\rm d}x\wedge {\rm d}y}{(1+x^2+y^2)^2}.
$$
As a consequence of our results and \cite[Theorem 4.4]{BBHLS}, this is the only symplectic structure on the projective space, up to a constant multiplicative factor, turning  (\ref{eqf})  into a LH system for arbitrary complex function $\beta(t)$, and real functions $\lambda_1(t)$ and $\lambda_2(t)$.
 

\sect{Planar $\mathfrak{h}_2$-LH systems}

 We now focus on the Lie algebra I$_{14A}\simeq \mathbb{R} \ltimes \mathbb{R}^r$ with $r=1$ of table~\ref{table1}, so with a basis of vector fields
 $X_1= {\partial_x}$, $X_2=  {\eta_1(x)\partial_y} $ with $\eta_1(x)\notin\langle 1\rangle$. If we require that these close a non-Abelian Lie algebra and we choose, with no loss of generality, that $[X_1,X_2]=X_2$, then 
 $\eta_1(x)={\rm e}^x $ up to an irrelevant proportional non-zero constant, that is
\be
 X_1= \frac{\partial}{\partial x},\qquad X_2={\rm e}^x \frac{\partial}{\partial y},\qquad [X_1,X_2]=X_2,
\label{h2}
\ee
and we denote  $\langle X_1,X_2\rangle:={\mathfrak h}_2\simeq \mathbb{R} \ltimes \mathbb{R}\simeq {\rm I}^{r=1}_{14A}$. This is a Vessiot--Guldberg Lie algebra of Hamiltonian vector fields relative to the symplectic form
$\omega=\dd x\wedge \dd y$. Hence, we can choose
\be
h_1= y,\qquad h_2=-{\rm e}^x,\qquad \{h_1,h_2\}_\omega=-h_2.
\label{hh2}
\ee

In the following, we show, as a new result,  that ${\mathfrak h}_2$ underlies  the complex Bernoulli differential equations with   $t$-dependent real coefficients and  we relate this result with other known   ones concerning generalised Buchdahl equations   and $t$-dependent Lotka--Volterra systems. It is remarkable that Cayley--Klein Riccati equations (\ref{CKRE2})  with $a_2(t)=0$ are $\mathfrak{h}_2$-LH systems as well. Finally, we prove that all of these systems belong to the same class  I$_{14A}^{r=1}\simeq {\mathfrak h}_2$.


\subsection{Complex Bernoulli differential equation with $t$-dependent real coefficients}

Let us restrict ourselves to studying the complex differential Bernoulli equation (\ref{complexBernoulli}) of the form
\begin{equation}\label{complexBernoulliReal}
 \frac{\dd z}{\dd t}=a_1(t)z+a_2(t)z^n,\qquad n\notin\{0,1\},
\end{equation}
with $a_1(t),a_2(t)$ being arbitrary {\em real} functions.  It can proved that the planar nonautonomous Bernoulli-like equations appearing in \cite[p.~197]{Muriel} when solving equations in the Abel chains can be mapped into a particular case of the above system by means of an appropriate diffeomorphism.

The equation (\ref{complexBernoulliReal})  can be studied in terms of the $t$-dependent vector field $X=a_1(t)X_0+a_2(t)X_2$,
where we recall that the vector fields $X_0$ and $X_2$ are given in (\ref{VectComBer}) and satisfy 
$$
[X_0,X_2]=(n-1)X_2,
$$
which  is isomorphic to   $\mathfrak{h}_2$. In the GKO classification~\cite{GKP92, BBHLS} there is just one Lie algebra isomorphic to $\mathfrak{h}_2$ whose vector fields are not proportional at each point: I$_{14A}$ with $r=1$. So,  $\langle X_0,X_2\rangle$ is a Lie algebra of Hamiltonian vector fields in view of the results of table~\ref{table1}. 

In order to study this system, we provide the change of variables 
$$
x:=\ln \left|\frac{r^{n-1}}{\sin[\theta(n-1)]}\right|,\qquad y:=-\frac{{\rm cotg}[\theta(n-1)]}{n-1}
$$
mapping the vector fields $\bar X_1=X_0/(n-1)$ and  $X_2$ into (\ref{h2}). The symplectic form and Hamiltonian functions for $  X_0$ and $X_2$ in the initial variables $r,\theta$ can be obtained from  $\omega=\dd x\wedge \dd y$ and  (\ref{hh2})  by inverting the above change of variables.


  \subsection{Generalized Buchdahl equations}

The generalized Buchdahl equations, appearing in the study of relativistic fluids~\cite{Bu64,CSL05} and whose properties have been studied through a Lagrangian approach in \cite{CN10},  are the    second-order differential equations given by
 $$
 \frac{\dd^2 x}{\dd t^2}=a(x)\left(\frac{\dd x}{\dd t}\right)^2+b(t)\frac{\dd x}{\dd t},
 $$
 for arbitrary functions $a(x)$ and $b(t)$. 
If we set   $y:= \dd x/\dd t $, we  find the first-order   system  of differential equations
  \begin{equation}\label{Buchdahl}
 \frac{\dd x}{\dd t}=y,\qquad
 \frac{\dd y}{\dd t}=a(x)y^2+b(t)y ,
 \end{equation}
 which is associated with the $t$-dependent vector field
 $X=X_2+b(t)X_1$,
 where
 \begin{equation}\label{VectBuch}\nonumber
 X_{1}=y\dfrac{\partial}{\partial y},\qquad X_2=y\frac{\partial }{\partial x}+a(x)y^2\frac{\partial}{\partial y}, 
 \end{equation}
satisfy  $ [X_1,X_2]=X_2$.
These vector fields span a Lie algebra diffeomorphic to  I$_{14A}^{r=1}\simeq {\mathfrak h}_2$ and  (\ref{Buchdahl}) becomes a LH system. The corresponding symplectic form and Hamiltonian functions can be found in~\cite{BBHLS}.


 \subsect{$t$-dependent Lotka--Volterra systems}

Finally, consider the particular Lotka--Volterra systems~\cite{Tr96,JHL05}  given by
  \begin{equation}\label{LV}
 \frac{\dd x}{\dd t}=ax-g(t)(x-ay)x,\qquad
 \frac{\dd y}{\dd t}=ay-g(t)(bx-y)y,\qquad a\ne 0,
\end{equation}
 where $g(t)$ determines the variation of the seasons, while  $a$ and $b$ are constants describing the interactions among the species.
System (\ref{LV}) is associated with  the $t$-dependent vector field $
X=X_1+g(t)X_2$
 where
 $$
 X_1=ax\frac{\partial}{\partial x}+ay\frac{\partial}{\partial y},\qquad X_2=-(x-ay)x\frac{\partial}{\partial x}-(bx-y)y\frac{\partial}{\partial y} ,
 $$
satisfy
$$
[X_1,X_2]=a X_2 .
$$
Hence, (\ref{LV}) is a Lie system. Moreover, it has been proven in~\cite{BBHLS} that, except for the case with $a=b=1$, this is also a LH system belonging to the family  I$_{14A}^{r=1}\simeq {\mathfrak h}_2$. The symplectic form and the Hamiltonian functions for $X_1$ and $X_2$ can be found in \cite{BBHLS}.

Hence, we conclude this section with the following statement.

\begin{proposition} The complex Bernoulli differential equation (\ref{complexBernoulliReal}), the generalized Buchdahl equations (\ref{Buchdahl}) and the
$t$-dependent Lotka--Volterra systems (\ref{LV}) (with the exception of $a=b=1$) are   LH systems with a Vessiot--Guldberg Lie algebra diffeomorphic to
 I$_{14A}^{r=1}\simeq {\mathfrak h}_2$   in  table~\ref{table1}. Thus all of these systems	 are locally diffeomorphic.
\end{proposition}


   \sect{Constants of motion and superposition rules}

  As commented in the introduction,  one of the most outstanding properties of LH systems is that their superposition rules (which exist for any Lie system) can be obtained by applying the coalgebra approach  recently introduced in~\cite{BCHLS13Ham} in an easier way than by applying traditional methods \cite{PW,CGM07}. 
    Essentially,  for any LH system,  this  procedure requires to endow the LH algebra with a coalgebra structure which is provided by a trivial (non-deformed) coproduct map.    Next,  $t$-independent  constants of motion can be obtained from the $m$th-order coproduct of a non-trivial  Casimir invariant and 
 the     corresponding superposition rule  can then be worked out by starting from such constants of motion.
 
 In this section we firstly provide the constants of motion for the LH algebras displayed in table~\ref{table1} and, secondly, we use them in the construction of superposition rules for planar LH systems. In this respect, 
 we recall that, to the best of our knowledge, this coalgebra approach has only been applied in~\cite{BCHLS13Ham} to  planar $\mathfrak{sl}(2)$ and $\mathfrak{so}(3)$-LH systems obtaining their  superposition rules. Therefore, we shall  restrict ourselves to studying the remaining LH systems here considered, that is, LH systems of classes P$_1$, P$_5$ and I$^{r=1}_{14A}$ (via the class I$_8$). Hereafter, we shall omit most technical details, which can be  found in~\cite{BCHLS13Ham},  and so we shall  briefly summarize the necessary essential tools so as to have a self-contained paper.


   \subsect{Constants of motion}

Assuming the notation introduced in section 3, we consider a Lie algebra 
$\mathfrak{g}$ spanned by the generators $\{ v_1,\dots,v_l \}$ and its corresponding symmetric algebra $S(\mathfrak{g})$  understood as a Poisson algebra.  Then $S(\mathfrak{g})$ can always be endowed with a  Poisson coalgebra structure by introducing the (non-deformed)
  coproduct  map   ${\Delta} : S(\mathfrak{g})\rightarrow
S(\mathfrak{g}) \otimes S(\mathfrak{g})$    defined by
 \begin{equation}\label{Con}
{\Delta}(v)=v\otimes 1+1\otimes v,\qquad \forall  v\in\mathfrak {g}\subset S(\mathfrak{g}),
\end{equation}
which is a Poisson algebra homomorphism.  
The coproduct $\Delta\equiv \Delta^{(2)}$ can be extended to a third-order coproduct  $\Delta^{(3)}: S(\mathfrak{g})\rightarrow
S(\mathfrak{g}) \otimes S(\mathfrak{g})\otimes S(\mathfrak{g})\equiv   S^{(3)}(\mathfrak{g})$ by means of  the coassociatity condition~\cite{CP}, $
\Delta^{(3)}:=(\Delta \otimes {\rm Id}) \circ \Delta=({\rm Id} \otimes \Delta) \circ \Delta$, namely
$$
{\Delta}^{(3)}(v)=v\otimes 1\otimes 1 +1\otimes v\otimes 1+1\otimes 1\otimes v,\qquad \forall  v\in\mathfrak {g}\subset S(\mathfrak{g}).
\label{3co}
$$
And  the {$m$th-order coproduct} map
$
\Delta ^{(m)}:  S(\mathfrak{g})\rightarrow   S^{(m)}(\mathfrak{g})$
can be defined, recursively, as  
$$
\label{copr}
{\Delta}^{(m)}:= ({\stackrel{(m-2)-{\rm times}}{\overbrace{{\rm
Id}\otimes\ldots\otimes{\rm Id}}}}\otimes {\Delta^{(2)}})\circ \Delta^{(m-1)},\qquad m\ge 3,
$$
which, clearly, is also  a Poisson algebra homomorphism.

Let $X$ be a LH system related to a LH algebra ${\cal{H}}_\Lambda$ spanned by the linearly independent Hamiltonian functions  $\{ h_1,\dots,h_l \}$ and let $D$ be  the Poisson algebra morphism $D:  S(\mathfrak{g}) \to C^\infty(M)$ induced by extending to $S(\mathfrak{g})$ the injection $\iota:\mathfrak{g}\hookrightarrow {\cal{H}}_\Lambda\subset C^\infty(M)$, with $\phi(v_i)=h_i$ for $i=1,\ldots,l$. By extension, we can construct a family of   Poisson algebra morphisms 
$D^{(m)}:  S^{(m)}(\mathfrak{g}) \to C^\infty(M)^{(m)}\subset  C^\infty(M^{m})$.  If $C$   is a polynomial Casimir  of the Poisson algebra $S(\mathfrak{g})$, say $C=C(v_1,\dots,v_l)$, then $D(C)$ is a constant of motion for $X$ and the functions 
\begin{equation}\label{invA}
F^{(k)}(h_1,\dots,h_l) = D^{(k)}\left[\Delta^{(k)} \left({C( v_1,\dots,v_l)} \right) \right] ,   \qquad k=2,\ldots,m,
\end{equation}
are  $t$-independent constants of motion  for the diagonal prolongation
$\widetilde X$ to the $m$th manifold $M^m$, namely if $X=\sum_{i=1}^nX^i(x)\partial/\partial x^i$, then
$$
\widetilde X=\sum_{a=1}^m\sum_{i=1}^nX^i(x_{(a)})\frac{\partial}{\partial x^i_{(a)}},
$$
where $(x_{(1)},\ldots,x_{(m)})\in M^m$ (see \cite{CGM07} for details on diagonal prolongations). Observe also that each $F^{(k)}$ can naturally be considered as
a function of $C^\infty(M^m)$ for every $m\geq k$.

If all the $F^{(k)}$ are non-constant functions, then  they form a set of  $(m-1)$      functionally independent functions  in involution in $C^\infty(M^m)$~(cf.~theorem 26 in \cite{BCHLS13Ham}).
Furthermore, from the functions $F^{(k)}$ other 
 constants of motion can be obtained in the form  
\begin{equation}\label{invB}
F_{ij}^{(k)}=S_{ij} ( F^{(k)}   ) , \qquad 1\le  i<j\le  k,\qquad k=2,\ldots,m,
\end{equation}
where $S_{ij}$ is the permutation of variables $x_{(i)}\leftrightarrow
x_{(j)}$ on $M^m$. Indeed, since $\widetilde X$ is invariant under the permutations $x_{(i)}\leftrightarrow x_{(j)}$, then the $F_{ij}^{(k)}$ are  also $t$-independent constants of motion  for the diagonal prolongations $\widetilde X$ to $M^m$. 

Let us now illustrate the previous procedure by considering a LH  system $X$ taking values in ${\rm P}_1\simeq  {\mathfrak{iso}}(2)$. Using the basis  $\{X_1,X_2,X_3\}$ of P$_1$   given in table~\ref{table1}, we have 
$$
[X_1, X_2] = 0,\qquad 	[X_1, X_3]=-X_2, \qquad 	[X_2, X_3] = X_1.
$$
The corresponding LH algebra is isomorphic to $\overline{\mathfrak{iso}}(2)$ and it admits a basis $\{ h_1,h_2,h_3,h_0\}$ (see table~\ref{table1}) satisfying commutation relations   
\be
\{h_1,h_2\}_\omega=h_0,\qquad \{h_1,h_3\}_\omega=h_2,\qquad \{h_2,h_3\}_\omega=-h_1,\qquad  \{h_0,\cdot\}_\omega=0 ,
\label{xa}
\ee
with respect to the  canonical symplectic structure $\omega={\rm d}x \wedge {\rm d}y$.
The symmetric Poisson algebra $S\left(\overline{\mathfrak{iso}}(2)\right)$  of  $\overline{\mathfrak{iso}}(2)$, where $\{ v_1,v_2,v_3,v_0\}$ is a basis of $\overline{\mathfrak{iso}}(2)$ fulfilling  the commutation rules ({\ref{xa}), 
has a non-trivial  Casimir  invariant given by 
$$
C=v_{3}v_{0}-\tfrac{1}{2}(v_{1}^{2}+v_{2}^{2}).
$$
From it, we obtain a trivial constant of motion on the variables $(x,y)\equiv (x_1,y_1)$:
\bea
&& F=D(C)= \phi(v_3)\phi(v_0)-\tfrac 12 \left( \phi^2(v_1) +\phi^2(v_2) \right) \nonumber\\
&& \quad =
h_{3}(x_1,y_1)h_{0}(x_1,y_1)-\tfrac{1}{2}[h_{1}^{2}(x_1,y_1)+h_{2}^{2}(x_1,y_1)]=\tfrac 12 (x_1^2+y_1^2)\times 1-
\tfrac 12 (y_1^2+x_1^2)=0 .
\nonumber
\eea
Nevertheless, once the coalgebra structure is introduced in $S\left(\overline{\mathfrak{iso}}(2)\right)$  through the coproduct (\ref{Con}), the functions (\ref{invA}) are no longer trivial ones and we find, for $m=3$ (so $k=2,3$),  that
$$
\begin{array}{rl}
  &F^{(2)} = D^{(2) } (\Delta(C) )
=  \left(h_3(x_1,y_1)+h_3(x_2,y_2)\right)\left(h_0(x_1,y_1)+h_0(x_2,y_2)\right) \nonumber\\   
&\qquad\qquad -\tfrac 12\left[(\left(h_1(x_1,y_1)+h_1(x_2,y_2)\right)^2+ \left(h_2(x_1,y_1)+h_2(x_2,y_2)\right)^2 \right]
\nonumber\\   
&\qquad
= \frac{1}{2}\left[
(x_{1}-x_{2})^{2}+(y_{1}-y_{2})^{2} \right] , \nonumber\\[2pt]
  &F^{(3)}  =D^{(3) } (\Delta(C) )
= \sum_{i=1}^3 h_3(x_i,y_i) \sum_{j=1}^3h_0(x_j,y_j)\nonumber\\
&\qquad\qquad
-\tfrac 12 \left[\left( \sum_{i=1}^3 h_1(x_i,y_i)   \right) ^2  + \left( \sum_{i=1}^3 h_2(x_i,y_i)   \right) ^2 \right]\nonumber\\
  &\qquad =\tfrac12  \sum_{1\le i<j}^3 \left[
(x_{i}-x_{j})^{2}+(y_{i}-y_{j})^{2} \right].
 \nonumber
  \label{intKS}
\end{array}
$$
For $k=2$, the constants (\ref{invB}) read
\bea
&& F_{12}^{(2)}=S_{12} ( F^{(2)}   ) \equiv  F^{(2)}   ,\qquad  F_{13}^{(2)}=S_{13} ( F^{(2)}   )= \tfrac{1}{2}\left[
(x_{3}-x_{2})^{2}+(y_{3}-y_{2})^{2} \right] ,\nonumber\\
&&  F_{23}^{(2)}=S_{23} ( F^{(2)}   )= \tfrac{1}{2}\left[
(x_{1}-x_{3})^{2}+(y_{1}-y_{3})^{2} \right]
\label{ff}
\eea
so that $ F^{(3)}= F_{12}^{(2)}+ F_{13}^{(2)}+ F_{23}^{(2)}$.  Observe that $ F^{(2)}$ and $ F^{(3)}$ satisfy that $\partial(F^{(2)},F^{(3)})/\partial_{(x_1,y_1)})\neq 0$. We can also choose two other functions among the set  $\{ F^{(2)},F_{13}^{(2)},F_{23}^{(2)}\}$ satisfying this condition. This will be important for deriving superposition rules for LH systems of class P$_1$.


\begin{table}[t] {\footnotesize
  \noindent
\caption{{\small Constants of motion for   LH algebras of table~\ref{table1}.   }}
\label{table2}
\medskip
\noindent\hfill
\begin{tabular}{ l l l }
\hline
 &&\\[-1.5ex]
\#&LH algebra &Casimir $C$ and invariants $F^{(k)}$\\[+1.0ex]
\hline
 &  & \\[-1.5ex]
P$_1$& $ \overline{ {\mathfrak{iso}} }(2)$& $ C=h_{3}h_{0}-\frac{1}{2}(h_{1}^{2}+h_{2}^{2})$   \\[+1.0ex]
 &  & $F=0\qquad F^{(2)}=\frac{1}{2}\left[
(x_{1}-x_{2})^{2}+(y_{1}-y_{2})^{2} \right]   $
  \\[+2.0ex]
P$_2$& $  { {\mathfrak{sl}} }(2)$& $ C=  h_{1}h_{3} -h_{2}^{2}$   \\ 
 &  & $F=1\qquad F^{(2)}= \dfrac{(x_{1}-x_{2})^{2}+(y_{1}+y_{2})^{2}}{y_{1}y_{2}}$
  \\[+2.0ex]
P$_3$& $   {\mathfrak{so}} (3)\oplus\mathbb{R} $& $ C= 4h_{1}^{2}+h_{2}^{2}+h_{3}^{2}+2 h_{1}h_{0}$   \\ 
 &  & $F=0\qquad F^{(2)}= -\dfrac{(x_{1}-x_{2})^{2}+(y_{1}-y_{2})^{2}}
{(1+x_{1}^{2}+y_{1}^{2})(1+x_{2}^{2}+y_{2}^{2})}$
  \\[+2.5ex]
P$_5$& ${\mathfrak h}_6\simeq \overline{\mathfrak{sl}(2 )\ltimes \mathbb{R}^2}$& $ C= 2 \left(h_{1}^{2}h_{5}-h_{2}^{2}h_{4}-h_{1}h_{2}h_{3} \right)-h_{0}(h_{3}^{2}+4 h_{4}h_{5})$   \\[2pt]
 &  & $F=0\qquad F^{(2)}= 0\qquad F^{(3)}=  \left[
x_{1}(y_{2}-y_{3})+x_{2}(y_{3}-y_{1})+x_{3}(y_{1}-y_{2}) \right]^{2} $
  \\[+2.0ex]
I$_4$& $  { {\mathfrak{sl}} }(2)$& $ C=  h_{1}h_{3} -h_{2}^{2}$   \\ 
 &  & $F=-\frac 1 4 \qquad F^{(2)}=- \dfrac{(x_{2}-y_{1})(x_{1}-y_{2})}{(x_{1}-y_{1})(x_{2}-y_{2})}$
  \\[+2.0ex]
I$_5$& $  { {\mathfrak{sl}} }(2)$& $ C=  h_{1}h_{3} -h_{2}^{2}$   \\ 
 &  & $F=0\qquad F^{(2)}= \dfrac{(x_{1}-x_{2})^{2}}{(2y_{1} y_{2})^2}$
  \\[+2.0ex]
I$_8$& $ \overline{ {\mathfrak{iso}} }(1,1)$& $ C=h_{1}h_{2}+h_{3}h_{0} $   \\[+1.0ex]
 &  & $F=0\qquad F^{(2)}= (x_{1}-x_{2})(y_{1}-y_{2})$
  \\[+2.0ex]
I$_{14A}^{r=2}$& $\mathbb{R} \ltimes \mathbb{R}^{2}$ & $ C=h_2 h_3\qquad    \eta_1(x)={\rm e}^x\qquad    \eta_2(x)={\rm e}^{-x}$   \\[+1.0ex]
 &  & $F=-1 \qquad F^{(2)}=-2\left[1+\cosh(x_1-x_2) \right]  $
  \\[+2.0ex]
I$_{14B}^{r=2}$& $\overline{\mathbb{R} \ltimes \mathbb{R}^{2} }$ & $ C= h_2^2+2h_3 h_0\qquad     \eta_2(x)=x  $   \\[+1.0ex]
 &  & $F= 0 \qquad F^{(2)}=-\left(x_1-x_2 \right)^2  $
  \\[+2.0ex]
I$_{16}$& $\overline{  {\mathfrak{h}_2 \ltimes  \mathbb{R}^{r+1}}} $ & $ C= \dfrac{2h_{2}^{3}+6 h_{2}h_{4}h_{0}+3 h_{5}h_{0}^{2}}
{3 h_{0}^{2}\left(h_{2}^{2}+2 h_{4}h_{0}\right)^{ {3}/{2}}}\qquad \ r\ge  2$   \\[+1.5ex]
 &  & $F=\mbox{indet} \qquad F^{(2)}= 0\qquad  F^{(3)}= \dfrac{
(x_{1}+x_{2}-2x_{3})(x_{1}+x_{3}-2x_{2})(x_{2}+x_{3}-2x_{1})
}
{54\sqrt{2} (x_{1}x_{2}+x_{1}x_{3}+x_{2}x_{3} -x_{1}^{2}-x_{2}^{2}-x_{3}^{2})^{3/2}}$
  \\[+3.0ex]
     \hline
\end{tabular}
\hfill}
\end{table}


We display in table~\ref{table2}  the first non-trivial invariants  $ F^{(k)}$ for each of the LH algebras of table~\ref{table1}, except for the trivial Abelian cases  ${\rm I}_1\simeq \mathbb R$ and  ${\rm I}_{12}\simeq \mathbb R^{r+1}$.  This is 
$ F^{(3)}$ for the classes  P$_5$ and I$_{16}$, and $ F^{(2)}$ for the remaining  ones. 
Notice that we have   expressed the corresponding Casimir $C$ in terms of $h_i$, instead of $v_i$, to facilitate the reading with respect to table~\ref{table1}.
We remark that the LH algebras of the classes I$_{14A}$,  I$_{14B}$ and  I$_{16}$ have no non-trivial invariant for $r=1$ and that for 
I$_{14A}$ and   I$_{14B}$ a choice of the functions $\eta_j(x)$ must be performed. We have worked out this latter case with $r=2$ and the specific functions $\eta_j(x)$ indicated in table~\ref{table2}.
In all the cases for which  $ F^{(2)}$ is not a trivial constant, it is found that
$F^{(3)}= F_{12}^{(2)}+ F_{13}^{(2)}+ F_{23}^{(2)}+{\rm constant}$.

The results of table~\ref{table2} are important due to the fact that allow us to construct a superposition rule for $X$. More precisely, let $X$ be a LH system on $\mathbb{R}^2$ admitting a Vessiot--Guldberg Lie algebra of Hamiltonian vector fields $V$ with basis $\{X_1,\ldots,X_l\}$. Then, a set $I_1,\ldots, I_n\in C^\infty(M^{m})$ of functionally independent constants of motion  for the diagonal prolongation $\widetilde X$ on $M^{m}$, with $m_0:=m-1$ being such that $\widetilde X_1\wedge\ldots\wedge \widetilde{X}_l\neq 0$ for the prolongations to $M^{m_0}$,  enables us to determine a superposition rule for $X$ provided $\partial (I_1,\ldots,I_n)/\partial(x^1_{(0)},\ldots,x^n_{(0)})\neq0 $  (see \cite{CGM07} for details). Indeed, the latter condition ensures that $x^1_{(0)},\ldots,x^n_{(0)}$ can be expressed in terms of $n$-constants $k_1,\ldots,k_n$ and the remaining variables in $M^m$ by solving the equations $I_1=k_1,\ldots,I_n=k_n$. This leads to a superposition rule for $X$. For Lie systems, $I_1,\ldots,I_n$ are usually obtained by solving a system of PDEs, whereas they can be derived algebraically for planar LH systems following table~\ref{table2}. Exemplifying this for several new superposition rules for LH systems on the plane will be the aim of following subsections.


   \subsect{Superposition rules for LH systems of class P$_1\simeq {\mathfrak{iso}(2)}$}

Consider a LH system $X$ with a Vessiot--Guldberg Lie algebra P$_1$. It can be proved that $m_0=2$  in this case. 
Let us consider the three constants of motion (\ref{ff}) written as 
\bea
&&  F^{(2)}= \tfrac{1}{2}\left[
(x_{1}-x_{2})^{2}+(y_{1}-y_{2})^{2} \right]=\tfrac 12 k^2_1\ge 0 ,\nonumber\\
&&  F_{23}^{(2)}=\tfrac{1}{2}\left[
(x_{1}-x_{3})^{2}+(y_{1}-y_{3})^{2} \right]= \tfrac 12 k^2_2\ge 0 ,\nonumber\\
&&  F_{13}^{(2)}=\tfrac{1}{2}\left[
(x_{3}-x_{2})^{2}+(y_{3}-y_{2})^{2} \right]= \tfrac 12 k^2_3 >0.
\label{super1}
\eea
We aim to express the general solution $(x_1(t),y_1(t))$ of our LH system in terms of two different particular solutions $(x_i(t),y_i(t))$, with $i=2,3$. This is obtained by starting from   the two first equations   (\ref{super1}). The resulting expressions can next be simplified by introducing the third constant $k_3$, giving rise to 
\bea
&&x_1^\pm(x_2,y_2,x_3,y_3,k_1,k_2,k_3)  = x_2+\frac{k_1^2+k_3^2-k_2^2}{2k_3^2}(x_3-x_2)\mp 2A\frac{(y_3-y_2)}{k^2_3},\nonumber\\
&&y_1^\pm(x_2,y_2,x_3,y_3,k_1,k_2,k_3)   =y_2+\frac{k_1^2+k_3^2-k_2^2}{2k^2_3}(y_3-y_2)\pm 2A\frac{(x_3-x_2)}{k^2_3} ,
\label{sup1}
\eea
where $A$  is a constant   given in terms of $k_1,k_2,k_3$ by
\be
A=\frac 1{4}\sqrt{2(k_1^2k_2^2+k_1^2k_3^2+k_2^2k_3^2)-(k_1^4+k_2^4+k_3^4)} .
\label{area}
\ee
We can assume with no loss of generality that $k_1,k_2,k_3\geq 0$. Since we assume $(x_2,y_2)\neq (x_3,y_3)$, we can set $k_3>0$.  
Therefore, the equations (\ref{sup1}) comprise {\em  two} cases (according to the signs $\pm$ and $\mp$),  which are well defined when 
 the radicand of $A$ is non-negative, namely when $k_i\leq k_j+k_l$ for different $i,j,l$.
  Since $k_3=k_3(x_2,y_2,x_3,y_3)$, we can understand expressions (\ref{sup1}) as a superposition rule $\Psi:(x_2,y_2,x_3,y_3,  \pm k_1,k_2)\in U\subset\mathbb{R}^4\times(\mathbb{R}\times\mathbb{R}_{k_2\geq 0} )\mapsto (x^\pm_1,y^\pm_1)\in \mathbb{R}^2$. This enables us to express the general solution $(x_1(t),y_1(t))\equiv (x(t),y(t))$ of $X$ in terms of two different particular solutions $(x_2(t),y_2(t))$ and $(x_3(t),y_3(t))$. 

Alternatively, the constants of motion (\ref{super1}) and the derivation of the superposition rule (\ref{sup1}) admit a geometrical approach. If
   $k_1,k_2,k_3$ are non-negative constants,    these   can be understood as the lengths of the segments $\overline{Q_1Q_2}$, $\overline{Q_1Q_3}$ and $\overline{Q_2Q_3}$ between the points $Q_1:=(x_1,y_1)$, $Q_2:=(x_2,y_2)$ and $Q_3:=(x_3,y_3)$ in $\mathbb{R}^2$ as shown in figure~\ref{figure1}. Hence,   
   the area of the triangle $Q_1Q_2Q_3$ is just the constant $A$ (\ref{area}), that is, the Heron's formula.
    From  figure~\ref{figure1}  we obtain that
\begin{equation}\nonumber
\begin{aligned}
x^\pm_1(x_2,y_2,x_3,y_3,k_1,k_2,k_3) &=x_2+k_1\cos(\alpha_1\pm\alpha_2)=x_2+k_1(\cos\alpha_1\cos\alpha_2\mp\sin\alpha_1\sin\alpha_2) ,\\
y^\pm_1(x_2,y_2,x_3,y_3,k_1,k_2,k_3) &=y_2+k_1\sin(\alpha_1\pm \alpha_2)=y_2+k_1(\sin\alpha_1\cos\alpha_2\pm \cos\alpha_1\sin\alpha_2),
\end{aligned}
\end{equation}
and by introducing the triangle relations
$$
\sin\alpha_2= \frac{2 A}{k_1k_3},\qquad \cos\alpha_2=\frac{k_1^2+k_3^2-k_2^2}{2k_1k_3},\qquad \sin\alpha_1=\frac{y_3-y_2}{k_3}, \qquad \cos\alpha_1=\frac{x_3-x_2}{k_3}, 
$$
we directly recover the equations (\ref{sup1}).


\begin{figure}[t]
\begin{center}
\begin{picture}(180,140)
\footnotesize{
\put(15,15){\makebox(0,0){$Q_2(x_2,y_2)$}}
\put(22,22){$\bullet$}
\put(127,127){\makebox(0,0){$Q_1(x_1,y_1)$}}
 \put(190,101){\makebox(0,0){$Q_3(x_3,y_3)$}}
 \put(187,65){\makebox(0,0){$y_3-y_2$}}
\put(165,93){$\bullet$}
  \put(85,15){$x_3-x_2$}
\put(68,80){\makebox(0,0){$k_1$}}
\put(117,80){\makebox(0,0){$k_3$}}
\put(148,113){\makebox(0,0){$k_2$}}
 \put(23,25){\line(1,0){145}}
\put(117,117){$\bullet$}
\put(165,22){$\bullet$}
\put(55,30){\makebox(0,0){$\alpha_1$}}
\put(55,46){\makebox(0,0){$\alpha_2$}}
 \qbezier(46,25)(44,35)(38,38)
\qbezier(25,25)(50,50)(120,120)
 \put(-2, 11){\qbezier(25,12.5)(90,45)(169,85)}
 \put(104,127){\qbezier(14,-7)(30,-15)(63,-31.5)}
 \put(167,25){\line(0,1){71}}
}
\end{picture}
\end{center}
\noindent
\\[-50pt]
\caption{\footnotesize Geometrical description of the three constants of motion (\ref{super1}) as the sides of a triangle.}
\label{figure1}
\end{figure}
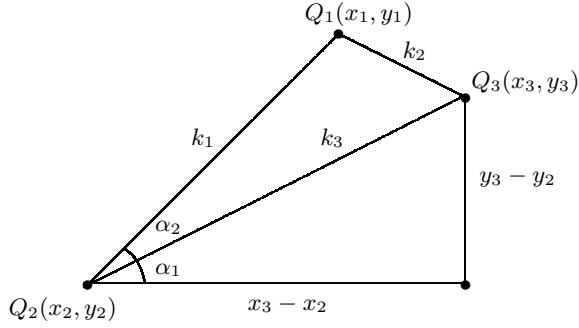


We remark that  the  result (\ref{sup1}) can be applied to all LH systems of class P$_1$. For instance, this can further be  used so as to obtain a superposition rule for the Bernoulli differential equations (\ref{PlanBer}) with $a_1^R(t)=0$, provided that a change of variables $(x,y)\to (r, \theta)$ mapping its Vessiot--Guldberg Lie algebra into P$_1$ is given. Equivalently, we can repeat the above procedure for the LH algebra of Bernoulli differential equations without deriving a diffeomorphism.


   \subsect{Superposition rules for LH systems of classes  I$_8\simeq {\mathfrak{iso}}(1,1)$ and  I$_{14A}^{r=1}\simeq \mathfrak{h}_2$}

   The class I$_{14A}^{r=1}\simeq \mathbb{R} \ltimes \mathbb{R}$ admits a LH algebra isomorphic to $\mathfrak{h}_2$ spanned by the functions given in table~\ref{table1} with commutation relations (\ref{hh2}).
   There does not exist a non-trivial Casimir  for $\mathfrak{h}_2$, so precluding, in principle, the obtention of a superposition rule through the coalgebra approach. Nevertheless, this problem can be circumvented by considering an inclusion of  I$_{14A}^{r=1}$ as a Lie subalgebra of a Lie algebra of another class admitting a LH algebra with  a non-trivial Casimir. 
   
   There are several classes containing I$_{14A}^{r=1}$,  e.g., P$_2$, I$_4$, I$_5$, I$_8$, P$_5$,\dots (cf.~see table 2 in~\cite{BBHLS}). Due to the simple form of the Casimir of the LH algebra $\overline{\mathfrak{iso}}(1,1)$ of I$_8$ and that 
     superposition rules for $\mathfrak{sl}(2)$-LH systems on the plane were already studied in~\cite{BCHLS13Ham}, we shall construct a new superposition rule for  LH systems of class I$_8\simeq {\mathfrak{iso}}(1,1)$, obtaining as a byproduct the one corresponding to LH ones of class I$_{14A}^{r=1}$.

   The   LH algebra $\overline{\mathfrak{iso}}(1,1)$ has commutation relations   
$$
\{h_1,h_2\}_\omega=h_0,\qquad \{h_1,h_3\}_\omega=-h_1,\qquad \{h_2,h_3\}_\omega=h_2,\qquad  \{h_0,\cdot\}_\omega=0,
$$
 with respect to $\omega={\rm d}x \wedge {\rm d}y$ in the basis $\{ h_1,h_2,h_3,h_0\}$ given in table~\ref{table1}. We have that $m_0=2$ for a LH system $X$ with Vessiot--Gulgdberg Lie algebra I$_8$. If $m=3$,  then
 we obtain from table~\ref{table2} three constants of motion for the diagonal prolongation $\widetilde X$ to $(\mathbb{R}^2)^3$ by applying (\ref{invB}): 
\bea
&&  F^{(2)}=  
(x_{1}-x_{2})   (y_{1}-y_{2}) =   k_1 ,\nonumber\\
&&  F_{23}^{(2)}=(x_{1}-x_{3})   (y_{1}-y_{3}) =   k_2 ,\nonumber\\
&&  F_{13}^{(2)}=(x_{3}-x_{2})   (y_{3}-y_{2}) =   k_3 .
\nonumber
\eea
Notice that $F^{(3)}= F_{12}^{(2)}+ F_{13}^{(2)}+ F_{23}^{(2)}$. In this case, $F^{(2)}=k_1,F^{(2)}_{23}=k_2$ can be understood as the equations on $\mathbb{R}^2$ of rectangular hyperbolas with centers $(x_2,y_2),(x_3,y_3)$.
Clearly, $F^{(2)}$ and $ F_{23}^{(2)}$ are functionally independent and allow us to express     $(x_1,y_1)$ in terms of $(x_2,y_2,x_3,y_3)$ and  $k_1,k_2$. The introduction of   $k_3$ again simplifies the final result  which reads
\bea
&&x_1(x_2,y_2,x_3,y_3,k_1,k_2,k_3)  = \frac 12(x_2+x_3) +\frac{k_2-k_1\pm B}{2(y_2-y_3)}   ,\nonumber\\
&&y_1(x_2,y_2,x_3,y_3,k_1,k_2,k_3)  =\frac 12(y_2+y_3) +\frac{k_2-k_1\mp B}{2(x_2-x_3)} ,
\label{sup2}
\eea
where
$$
B= \sqrt{ k_1^2+k_2^2+k_3^2-2(k_1k_2+k_1k_3+k_2k_3) }.
$$
Consequently, if we recall that $k_3=k_3(x_2,y_2,x_3,y_3)$, we obtain that (\ref{sup2}) leads to a  superposition rule for LH systems with a Vessiot--Guldberg Lie algebra I$_{14A}^{r=1}$  provided 
$$
 k_1^2+k_2^2+k_3^2\geq 2(k_1k_2+k_1k_3+k_2k_3).
 $$
 Let us now recover the superposition rule for a LH system of class I$_{14A}^{r=1}$ out of these results. It is immediate that up to a trivial change of variables I$_{14A}^{r=1}=\langle \partial_u,{\rm e}^u\partial_v\rangle$. Then, the change of variables 
   $$
   y={\rm e}^u,\qquad x=v{\rm e}^{-u},
   $$
   maps the basis of I$_{14A}^{r=1}$ into vector fields $-X_3$ and $X_1$ of I$_8$ given in table~\ref{table1}. Hence, every LH system of class I$_{14A}^{r=1}$ can be considered as a LH system of class I$_8$ and the above superposition rule for I$_8$ also applies, up to a change of variables, to LH systems of class I$_{14A}^{r=1}$.
It is worth noting that this  can be applied to the specific systems of class  I$_{14A}^{r=1}$ studied in section 9.


   \subsect{Superposition rules for two-photon LH systems}

Finally, consider a LH system with a Vessiot--Guldberg Lie algebra P$_5$. It can be proved that the prolongations of any basis of P$_5$ become linearly independent at a generic point for $m_0=3$. So we fix $m=4$.
The associated LH algebra $\mathfrak{h}_6$ with the basis given in table~\ref{table1} fulfills the commutation relations 
(\ref{twophoton}). According to table~\ref{table2},  the constant of motion $F^{(2)}=0$.
From $F^{(3)}$ written  in table~\ref{table2} and by using  (\ref{invB}) we obtain four constants of motion given by
\bea
&&  F^{(3)}=  
 \left(
x_{1}(y_{2}-y_{3})+x_{2}(y_{3}-y_{1})+x_{3}(y_{1}-y_{2}) \right)^{2}   =k_1^2,\nonumber\\
&& F_{34}^{(3)}=  \left(
x_{1}(y_{2}-y_{4})+x_{2}(y_{4}-y_{1})+x_{4}(y_{1}-y_{2}) \right)^{2}   =k_2^2 ,\nonumber\\
&& F_{24}^{(3)}=   \left(
x_{1}(y_{3}-y_{4})+x_{3}(y_{4}-y_{1})+x_{4}(y_{1}-y_{3}) \right)^{2}  =k_3^2 ,\nonumber\\
&& F_{14}^{(3)}=  \left(
x_{2}(y_{3}-y_{4})+x_{3}(y_{4}-y_{2})+x_{4}(y_{2}-y_{3}) \right)^{2} =k_4^2 ,
\nonumber
\eea
such that $ F^{(4)}=  F^{(3)}+   F_{34}^{(3)}+F_{24}^{(3)}+F_{14}^{(3)}  $. The constants of motion $F^{(3)}$ and $ F^{(4)}$ are functionally independent functions, but also 
$F^{(3)}$ and $ F_{34}^{(3)}$ which have a simpler form. Thus from the latter and by taking positive square roots we 
  find    $(x_1,y_1)$ written in terms of $(x_2,y_2,x_3,y_3,x_4,y_4)$ and the two constants $k_1$ and $k_2$. The result is rather simplified by introducing a third constant $k_4$ coming from the  positive square root of $ F_{14}^{(3)}$ (the remaining constant $k_3$ does not enter), yielding
 \bea
&&x_1(x_2,y_2,x_3,y_3,x_4,y_4,k_1,k_2,k_4)  =\left( 1+\frac{k_2-k_1}{k_4}\right) x_2-\frac{k_2}{k_4}  \, x_3  +\frac{k_1}{k_4}  \, x_4   ,\nonumber\\
&&y_1(x_2,y_2,x_3,y_3,x_4,y_4,k_1,k_2,k_4)  =\left( 1+\frac{k_2-k_1}{k_4}\right) y_2-\frac{k_2}{k_4}  \, y_3  +\frac{k_1}{k_4}  \, y_4   ,
\nonumber
\eea
which are well defined whenever $k_4\ne 0$. The above expression gives rise to a superposition rule for LH systems with a Vessiot--Guldberg Lie algebra P$_5$ by considering $k_4$ as a function   $k_3=k_3(x_2,y_2,x_3,y_3,x_4,y_4)$. Obviously, this result is also valid, up to an appropriate change of variables, to any other Lie system of class P$_5$. In particular, this result can be applied to the two-photon LH systems described in section 7, e.g., to $t$-dependent dissipative harmonic oscillators.


   \sect{Concluding remarks}
This work introduces the so-called Casimir tensor fields for certain finite-dimensional Lie algebras of vector fields. This allowed us to easily classify Lie algebras of vector fields on $\mathbb{R}^2$ isomorphic to $\mathfrak{sl}(2)$. In the future, we aim to extend our techniques to arbitrary finite-dimensional Lie algebras of vector fields.
We also hope to devise methods to construct and to classify general LH systems in a systematic way.

We have presented new LH systems: different kinds of  Bernoulli differential equations,  Cayley--Klein Riccati equations, planar diffusion Riccati systems, etc. We have related them with some already known LH systems scattered in the literature. These results are summarized in table~\ref{table3} where, according to the local classification of planar LH systems displayed  in table~\ref{table1}, we present the specific LH systems which are locally diffeomorphic within each class. 
For the sake of completeness,  we also indicate LH systems that have not been studied here, but that can be found in~\cite{BBHLS,CLS12Ham,BCHLS13Ham,CLS12,ADR12}.
As a result, table \ref{table3} details all LH systems on the plane with physical and mathematical applications appearing in the literature. Of course, the search of new applications of LH systems on the plane is still an open problem.

Furthermore, we have derived $t$-independent constants of motion for most of the planar LH algebras by applying a coalgebra approach. This has been used to derive new superposition rules in an algebraic manner. In this respect, we remark that this procedure makes use of the non-deformed coproduct map (\ref{Con}). This fact naturally suggests us trying to  extend such an approach to quantum (Poisson) algebras by considering deformed coproducts and deformed commutation rules. Thus the quantum deformation parameter $q$ would enter in the `deformed' (generalized) LH systems in such a manner that the initial systems would be recovered under the non-deformed limit $q\to 1$.

Observe that the Casimir function I$_{16}$ is the only element of table~\ref{table2} that is not an element of $S(\overline{  {\mathfrak{h}_2 \ltimes  \mathbb{R}^{r+1}}})$. Therefore it cannot be straightforwardly employed through the theory of this work. Although our methods can be generalized by using the approach given in \cite{CGLS} for this type of Casimir elements, this approach does not provide any significant improvement for planar LH systems and it will left for studying problems where non-polynomial Casimir functions will be the rule rather than the exception.

 Work on these lines is currently in progress.


\begin{table}[t] {\footnotesize
  \noindent
\caption{{\small Specific LH systems on the plane according to their class given in   table~\ref{table1}.   All of these systems have   $t$-dependent real coefficients except for P$_1$. The systems marked with `$\ast$' (I$_{14A}^{r=2}$ and  I$_{14B}^{r=2}$) have been studied in~\cite{BBHLS}, while the one 
  marked with `$\dagger$'  in P$_3$   can be found in~\cite{BCHLS13Ham,ADR12}.}}
\label{table3}
\medskip
\noindent\hfill
\begin{tabular}{ l   l l }
\hline
 & &\\[-1.5ex]
\# &  LH algebra & LH systems \\[+1.0ex]
\hline
 &    &\\[-1.5ex]
 P$_1$& $\overline{\mathfrak{iso}(2 )}$&  Complex Bernoulli equation $\dot z={\rm i}a(t)z+b(t)z^n$ for real $a(t)$ and complex $b(t)$  \\[+1.5ex]
P$_2$& $\mathfrak{sl}(2 )$&  Complex Riccati equation     \\[2pt]
 &  &Milne--Pinney  and Kummer--Schwarz   equations with $c>0$   \\[+1.5ex]
 P$_3$&   ${\mathfrak{so}}(3 )\oplus\mathbb{R},\, \mathfrak{so}(3) $& Projective Schr\"odinger equations on $\mathbb{C}\mathbb{P}^1$\\[2pt]
 & & Planar system with trigonometric nonlinearities$^\dagger$\\[+1.5ex]
  P$_5$& ${\mathfrak h}_6\simeq \overline{\mathfrak{sl}(2 )\ltimes \mathbb{R}^2}$&  Dissipative/damped harmonic oscillators, particle under specific electric fields \\[2pt]
   &  & Hamilton equations for quadratic Hamiltonians\\[2pt]
  &  & Second-order Riccati equation in Hamiltonian form\\[+1.5ex]
  I$_4$&  $\mathfrak{sl}(2 )$& Split-complex Riccati equation  \\[2pt]
  &  &  Coupled Riccati equations \\[2pt]
   &  &Milne--Pinney  and Kummer--Schwarz   equations with $c<0$ \\[2pt]
    &  &Planar diffusion Riccati system for $c_0=1$\\[+1.5ex]
    I$_5$& $\mathfrak{sl}(2 )$& Dual-Study   Riccati equation    \\[2pt]
 &  & Milne--Pinney  and Kummer--Schwarz   equations with $c=0$\\[2pt]
  &  & Harmonic oscillator \\[+2pt]
  &  &Planar diffusion Riccati system for $c_0=0$\\[+1.5ex]
   I$_{14A}^{r=1}$& ${\mathfrak h}_2\simeq \mathbb{R} \ltimes \mathbb{R}$&  Complex Bernoulli equation  $\dot z= a_1(t)z+a_2(t)z^n$ \\[2pt]
 &  & Generalised Buchdahl equations  \\[2pt]
  &  & Lotka--Volterra systems  \\[+1.5ex]
    I$_{14A}^{r=2}$& $  \mathbb{R} \ltimes \mathbb{R}^2$&  Quadratic polynomial   systems  $\dot x= bx+c(t)y+f(t)y^2$, $\dot y=y$ with  $b\notin \{1,2\}${}$^\ast$   \\[+1.5ex]
     I$_{14B}^{r=2}$& $ \overline{ \mathbb{R} \ltimes \mathbb{R}^2}$&  Quadratic polynomial   systems  $\dot x= bx+c(t)y+f(t)y^2$, $\dot y=y$ with  $b\in \{1,2\}${}$^\ast$\\[2pt]
  &  & A primitive model of viral infection$^\ast$ \\[+1.5ex]
    \hline
\end{tabular}
\hfill}
\end{table}




\section*{Acknowledgments}

The research of A.~Blasco and F.J.~Herranz was partially   supported by the Spanish Ministerio de Econom{\'{\i}}a y Competitividad     (MINECO) under grant MTM2013-43820-P and   by   Junta de Castilla y Le\'on  under   grant BU278U14.     J.~de Lucas and C.~Sard\'on acknowledge funding from the Polish National Science Centre  under the grant HARMONIA
DEC-2012/04/M/ST1/00523. C.~Sard\'on also acknowledges a fellowship provided by the University of Salamanca.


\end{document}